\documentclass[12pt]{article}

\pdfoutput=1

\usepackage{mathtools}
\usepackage[margin=1.25in]{geometry}
\usepackage{amsthm,amsmath,amssymb}
\usepackage[round]{natbib}
\usepackage[normalem]{ulem}
\usepackage{xcolor}
\usepackage{setspace}
\DeclareFontFamily{U}  {MnSymbolC}{}
\DeclareFontShape{U}{MnSymbolC}{m}{n}{
    <-6>  MnSymbolC5
   <6-7>  MnSymbolC6
   <7-8>  MnSymbolC7
   <8-9>  MnSymbolC8
   <9-10> MnSymbolC9
  <10-12> MnSymbolC10
  <12->   MnSymbolC12}{}
\DeclareFontShape{U}{MnSymbolC}{b}{n}{
    <-6>  MnSymbolC-Bold5
   <6-7>  MnSymbolC-Bold6
   <7-8>  MnSymbolC-Bold7
   <8-9>  MnSymbolC-Bold8
   <9-10> MnSymbolC-Bold9
  <10-12> MnSymbolC-Bold10
  <12->   MnSymbolC-Bold12}{}
\DeclareSymbolFont{MnSyC}         {U}  {MnSymbolC}{m}{n}
\DeclareMathSymbol{\diamonddot}{\mathbin}{MnSyC}{"7E}

\allowdisplaybreaks

\usepackage[nottoc]{tocbibind}

\usepackage{fnbreak} 

\usepackage[
  bookmarks=true,
  bookmarksnumbered=true,
  bookmarksopen=true,
  pdfborder={0 0 0},
]{hyperref}
\usepackage[capitalize]{cleveref}

\hypersetup{
  pdfauthor      = {Yannai A. Gonczarowski <yannai@gonch.name> and Yoram Moses <moses@technion.ac.il>},
  pdftitle       = {Common Knowledge, Regained},
}

\newtheorem{lemma}{Lemma}
\newtheorem{theorem}{Theorem}
\newtheorem{corollary}{Corollary}

\theoremstyle{definition}
\newtheorem{definition}{Definition}
\newtheorem{example}{Example}

\newcommand{\eqdef}{\triangleq}
\DeclareMathOperator{\supp}{supp}

\newcommand{\NN}{\mathbb{N}}
\newcommand{\Time}{T}
\newcommand{\histories}{\Omega}
\newcommand{\history}{\omega}
\newcommand{\allplayers}{N}
\newcommand{\players}{I}
\newcommand{\points}{\bar{\Omega}}
\newcommand{\partition}{\bar{P}}

\newcommand{\sometime}{\diamonddot}
\newcommand{\always}{\boxdot}
\newcommand{\subat}{{\scriptscriptstyle@}}

\title{Common Knowledge, Regained\thanks{The authors thank Bob Aumann, Drew Fudenberg, Joe Halpern, Sergiu Hart, Giacomo Lanzani, Shengwu Li, Jacob Leshno, Eric Maskin, Stephen Morris, Ellen Muir, Tomasz Strzalecki, Alex Wolitzky, and participants at the Harvard--MIT Economic Theory seminar, for insightful comments and discussions.}}
\author{
Yannai A. Gonczarowski\thanks{Department of Economics and Department of Computer Science, Harvard University | \emph{E-mail}: \mbox{\href{mailto:yannai@gonch.name}{yannai@gonch.name}}.}
\and
Yoram Moses\thanks{Department of Electrical and Computer Engineering, Technion---Israel Institute of Technology | \emph{E-mail}: \href{mailto:moses@technion.ac.il}{moses@technion.ac.il}.}
}
\date{April 23, 2024}

\begin{document}

\begin{titlepage}

\maketitle


\begin{abstract}
For common knowledge to arise in dynamic settings, all players must simultaneously come to know it has arisen. Consequently, common knowledge cannot arise in many realistic settings with timing frictions. This counterintuitive observation of \citet{HalpernM90} was discussed by \citet{ArrowAK87} and \citet{Aumann89}, was called a \emph{paradox} by \citet{Morris14}, and has evaded satisfactory resolution for four decades. We resolve this paradox by proposing a new definition for common knowledge, which coincides with the traditional one in static settings but is more permissive in dynamic settings. Under our definition, common knowledge can arise without simultaneity, particularly  in canonical examples of the Haplern--Moses paradox. We demonstrate its usefulness by deriving for it an agreement theorem \textit{{\`a}~la} \citet{Aumann76}, showing it arises in the setting of \citet{GeanakoplosP82} with timing frictions added, and applying it to characterize equilibrium behavior in a dynamic coordination game.
\end{abstract}

\thispagestyle{empty}

\end{titlepage}

\thispagestyle{empty}

\tableofcontents

\clearpage

\setcounter{page}{1}

{\raggedleft\small
\emph{``Knowledge equals profit.''}

\footnotesize ---Ferengi Rule of Acquisition \#74, Star Trek

\vspace{-1em}
}


\section{Introduction}

Economic theory makes abundant use of common knowledge assumptions. It is therefore of interest to understand how such common knowledge is obtained. As it turns out, in settings where the players' knowledge  evolves over time, attaining common knowledge may be quite a harsh requirement. Indeed, it has long been established that in the presence of rather mild timing frictions, common knowledge cannot dynamically arise. For example, suppose that a player~$\alpha$ sends a message containing either ``yes'' or ``no'' to another player~$\beta$. If the sending time of the message is unknown to~$\beta$ and the message might take either one or two seconds to arrive, then the content of the message might never become common knowledge among these two players.\footnote{Extremely briefly, two seconds after the message is sent, $\alpha$ knows that $\beta$ knows the content, but since $\beta$ might have received it just then and would consider it possible that the message was sent just one rather than two seconds before that, $\alpha$ does not know that $\beta$ knows that $\alpha$ knows that $\beta$ knows the content. A second later this is known, but $\alpha$ does not know that $\beta$ knows it to be so, etc. At no time is it true that $\alpha$ knows that $\beta$ knows that (repeating arbitrarily many times) $\alpha$ knows that $\beta$ knows the content. We review this issue in greater detail in \cref{prelims}.} This phenomenon has been termed by \citet{ArrowAK87} and by \citet{Aumann89} the ``Halpern--Moses Problem'' after the paper by \citet{HalpernM90} in which it was uncovered, and was called a ``paradox'' by \citet{Morris14}. Combining these two names, we henceforth refer to it as the ``Halpern--Moses Paradox.''\footnote{The example above, albeit with $0$ and $\varepsilon$ seconds rather than $1$ and $2$ seconds, is due to \citet{HalpernM90}, who use it to demonstrate that unless clocks are perfectly synchronized and messages are dated, common knowledge need not arise. A similar example was utilized by \citet{SteinerS11} to furthermore show that undated communication can destroy (probabilistic) common learning even if such learning would have held absent communication.}

The seminal paper of \citet{GeanakoplosP82} (see also \citealp{Geanakoplos94}) contains a beautiful analysis of how many messages should be sent back and forth between two players for common knowledge to arise. In fact, as we observe in \cref{prelims}, the Halpern--Moses Paradox applies also to such a setting in the presence of even slight timing frictions---i.e., uncertainty about delivery times and possibly non-synchronized watches---whether on the order of seconds, \emph{milli}seconds, or less.\footnote{For a review of the importance of timing frictions in the economic literature, see Section~3 of \citet{Morris14}.} Indeed,
even if player~$\alpha$ and player~$\beta$ exchange many back-and-forth messages, as long as there are timing frictions, new common knowledge---e.g., of the content of the first message---cannot arise. At the heart of all of these examples lies the observation by \citet{HalpernM90} that common knowledge can only arise if it becomes known simultaneously by all players, and since true simultaneity does not exist in real-life settings, new common knowledge cannot arise. Quite disturbingly, one consequence is that because of internet latency fluctuations during an online video meeting, at the conclusion of the meeting none of the new ideas raised in it are common knowledge. Moreover, given that there are small and uncertain neural information processing delays in human perception (not to mention quantum-theoretical uncertainty in the delivery of any form of information), even when two players look each other in the eye while shaking hands to seal a deal, this does not render the deal common knowledge.

In all of the above settings, even though common knowledge is never formally attained but only various approximations of common knowledge hold (e.g., $k$-level knowledge for high $k$, or various probabilistic approximations), it intuitively seems as if common knowledge \emph{effectively} holds: Indeed, it is hard to argue why after a few seconds of an online video meeting any player should not behave \emph{as if} there is common knowledge of words that were said at the beginning of the meeting. Similarly, in the \citet{GeanakoplosP82} setting with added timing frictions, the posteriors of the players do become identical (see \cref{sec:agreement}), as would have been the case had common knowledge arisen.\footnote{These contrast with the well-known electronic mail game of \citet{Rubinstein89}, in which even though various approximations of common knowledge hold, common knowledge \emph{does not even effectively exist}, in the sense that consequences of common knowledge fail to formally hold. The kind of problem that is manifested 
in the Halpern--Moses Paradox does not arise in the electronic mail game, since common knowledge neither formally nor effectively holds there.} Even though this discrepancy between the intuitive and formal meanings of common knowledge has been known and discussed for four decades, it has remained unresolved.

In this paper we offer a new definition for common knowledge, which overcomes the problematic aspects of the standard definition that are uncovered by the Halpern--Moses Paradox. Our definition coincides with the traditional one in static settings (settings in which knowledge does not evolve over time), but diverges from it and is more permissive in dynamic settings. Under our definition, common knowledge does hold in many settings in which ``common knowledge should effectively hold'' (but might not formally hold under the traditional definition), including all of the above examples. Despite being more permissive, common knowledge under our definition has similar desirable implications to those of traditionally defined common knowledge. In a precise sense, our definition resolves the Halpern--Moses Paradox.

The traditional definition, in a static setting, of a fact $\phi$ being common knowledge requires a strong form of mutuality between the epistemic states of the players: It requires that $\alpha$ knows that $\beta$ knows that $\alpha$ knows \ldots\ that $\phi$, for any level of nesting. A formal implication of this definition is that whenever 
$\phi$ is common knowledge, all players must know this (while when $\phi$ is not common knowledge, no player can know that it is).
The standard and straightforward way to generalize this definition to a dynamic setting requires precisely the same mutuality among all players, and implies that it must become known simultaneously by all players at some given instant. 
Our definition generalizes the static definition to a dynamic setting in a more nuanced way that still requires the same strong epistemic mutuality, while relaxing the simultaneity requirement.
Specifically, it reasons about each player's knowledge as a certain \emph{local event} occurs. E.g., it requires that \underline{$\alpha$ as she sends the message} knows that \underline{$\beta$ as she receives the message} knows that \underline{$\alpha$ as she sends the message} knows \ldots\ that $\phi$, for any level of nesting. Defined this way, common knowledge can arise in the absence of simultaneity.
In particular, we show in \cref{attaining} that it arises in each of the above examples, including the above-discussed \citet{GeanakoplosP82} setting with added timing frictions (even though traditionally defined common knowledge is formally not attainable in this setting due to these frictions).

Our definition, in addition to holding in settings in which the traditional definition has been criticized within economics and computer science for failing to hold, also holds true in formalizations of settings in which influential papers in other disciplines assert that common knowledge should hold. In a seminal paper on common knowledge in natural discourse, \citet{ClarkM81} identify \emph{co-presence} as a basis for establishing that people have common knowledge of facts of interest in practical settings. E.g., when $\alpha$ and $\beta$ meet at the cafe, their shared presence at the same table is the basis for concluding that they have common knowledge of their meeting.  
A closer inspection shows, however, similarly to the above discussion, that since the first instant at which~$\alpha$ observes $\beta$ at the table might not be truly simultaneous with the first instant at which $\beta$ observes this, the Halpern--Moses Paradox implies that $\alpha$ and $\beta$ do not obtain common knowledge of their meeting. We take the position that the problem here lies not in the claim that co-presence implies common knowledge but rather in the unintended feature of the traditional definition that requires simultaneity for common knowledge to arise dynamically.
In other words, we agree that co-presence \emph{should} give rise to common knowledge, once the latter is appropriately defined. It indeed does so under our new definition.

To make our notion of common knowledge more useful for analysis, it is desirable to have a convenient way to prove that it holds.
In \cref{induction} we provide an Induction Rule for common knowledge under our definition---a condition that is easy to check in many settings, and which can be used to establish that common knowledge holds.\footnote{The term Induction Rule is inspired by terminology used in the analysis of (traditionally defined) common knowledge in \citet{HalpernM92}; see also \citet{ClarkM81}.}
To ascertain the emergence of common knowledge (under our definition) using this Induction Rule, one need only check a simple condition regarding each player and apply the rule. This Induction Rule is at the heart of the proofs of many of our results.

The Halpern--Moses Paradox is often illustrated in settings where the temporal uncertainty is very small, which emphasizes its paradoxical nature. This might make it seem like the discussion of the emergence of common knowledge is limited to situations in which there is at least approximate simultaneity, or some form of co-presence of the different players.
The essence of common knowledge, as unearthed by our definition, however, turns out not to be limited in this way.
Imagine a note sent by carrier pigeon from one city to another, taking between $7$ to $8$ hours to arrive. There is no semblance of simultaneity between when it is sent and when it arrives, and certainly no co-presence between the sender and the receiver, and yet (assuming that pigeons always arrive safely) the content of the note carried by the pigeon is common knowledge (according to our definition) between the sender as she sends the pigeon off and the recipient as she receives it: The sender, as she sends the pigeon, knows that the receiver, as the pigeon arrives, will know that the sender, as she sent the pigeon knew that the receiver etc. And, this allows the sender and receiver to act based on having common knowledge of the message.\footnote{To take this example to the extreme, let us recall the \href{https://www.youtube.com/watch?v=iA1LmWWYxWY}{\textcolor{blue}{\uline{final scene}}} from the movie Back to the Future II. In this scene, Doc Brown is inside his DeLorean time machine when it is struck by lightning and disappears. Marty McFly, who was outside the car as the lightning struck, is terribly worried. Not a moment passes and a Western Union carrier arrives and delivers to McFly a letter that was in Western Union possession for 70 years, ever since it was sent by Doc Brown from the Old West, whence the lightning that struck his DeLorean sent him. Doc Brown gave explicit instructions for the letter to be delivered to McFly at that exact location, at that exact minute in time. The content of the letter is thus common knowledge (according to our definition) between Doc Brown as he sends the letter in 1885 and Marty McFly as he receives it in 1955. This is true despite there being no simultaneity whatsoever in this scenario: Doc Brown dies many decades before McFly receives the letter. Hence, not only do Doc and Marty fail to share co-presence as they attain common knowledge of the content, they do not even share ``co-existence''! Nonetheless, there certainly is some flavor of common knowledge here, which is captured by our definition.} Thus, it turns out that co-presence, which is in many settings seen as not only sufficient but also necessary in order for common knowledge to dynamically arise, is not strictly required. As we show in \cref{sec:cooccurrence}, co-presence can be relaxed to what we term \emph{co-occurrence}---the property of two events always (i.e., in any given history of the world) either both occurring, possibly at different times, or neither ever occurring. We formally define co-occurrence of events, and use it to define a necessary and sufficient condition for common knowledge to arise under our definition. Co-occurrence is not affected by timing frictions. Hence, using it to characterize common knowledge (under our definition) further showcases the robustness of our notion of common knowledge to timing frictions, regardless of their intensity.

{
\advance\leftmargini -.25em
The point of departure of what has become known as the \emph{Wilson Doctrine} is the famous excerpt from \citet[the emphasis is ours]{Wilson87}:
\begin{quote}\small
Game theory has a great advantage in explicitly analyzing the consequences of \textbf{trading rules that presumably are really common knowledge}; it is deficient to the extent it assumes other features to be common knowledge, such as one agent’s probability assessment about another’s preferences or information\ldots
\end{quote}
Indeed, many economic and game theoretic analyses make use of the idea that once the rules of some mechanism (or game, or contract) have been announced, they are common knowledge. A closer inspection reveals that in realistic settings this might never be the case \textit{vis a vis} the traditional definition of common knowledge, as even slight timing frictions foil the ability to harness common knowledge in the analysis. Under our new definition, common knowledge in such circumstances is regained, resolving a paradox that has baffled economists and computer scientists for four decades, and rendering \citeauthor{Wilson87}'s ``presumption'' (to use his own word) of common knowledge of the rules---and with it, all of the consequences that it entails---a precise mathematical truth.}

Common knowledge indeed has many appealing implications (see, e.g., \citealp{Aumann76,MilgromS82,BrandenburgerD87,AumannB95,Chwe99}) but the Halpern--Moses Paradox renders it impoverished or vacuous in various settings under its traditional definition. Under our new definition, common knowledge becomes non-vacuous in many of these settings. Despite our definition being more permissive, in \cref{using} we demonstrate that it is still powerful enough for deriving appealing implications traditionally associated with common knowledge, and in particular such implications that are known not to follow from any finite level of nested knowledge. We provide an agreement theorem (in the spirit of \citealp{Aumann76}) for our definition, and apply it to recover agreement on posteriors in the above-discussed \citet{GeanakoplosP82} setting with added timing frictions.

Common knowledge under our definition is not only a more permissive and more realistic sufficient condition for many economic implications to hold true, but it also characterizes equilibrium behavior in many situations. To see this, in the spirit of various ``coordinated attack'' games involving generals \citep[e.g.,][]{Rubinstein89,MorrisS97}, imagine two generals at two separate camps, who wish to simultaneously attack an enemy city at sundown if conditions are ripe. In order to attack, each general must prepare for attacking at least one hour in advance. Preparing without eventually attacking is costly, as is attacking alone or when conditions are not ripe. Around noon, one of the generals obtains information regarding whether or not conditions will be ripe for an attack at sundown. This general then sends an emissary to the other general (at the other camp) with this information. While the emissary is guaranteed to arrive at the other camp, this may take between one and two hours. Neither general has an accurate clock (each can only accurately recognize sunrise, noon, and sundown).

Preparing to attack once common knowledge, traditionally defined, of favorable conditions is attained misses out on the utility of attacking: Under the above conditions, this does not occur before sundown, at which time it is too late to prepare for attacking. However, it is an equilibrium for each general to prepare to attack once it is common knowledge, as defined in this paper, that conditions are favorable: The content of the message is common knowledge between the first general as she sends the emissary and the second general as she receives the message from the emissary, both of which are guaranteed to occur sufficiently early before sundown to allow proper preparation for attacking.

One might rightfully claim that an intricate epistemic analysis is hardly needed to identify that for the first general to prepare to attack once she sends the emissary, and for the second general to prepare to attack once the emissary arrives, constitutes an equilibrium. Imagine, though, analyzing the same game without being explicitly provided with details of how the generals learn about the world and communicate. In such a case it might be considerably less obvious whether an equilibrium in which both generals attack exists, and if so what it might be. Nonetheless, an epistemic characterization based on our notion of common knowledge holds true regardless of the specific technology by which the generals communicate and learn about whether conditions for attack are favorable. If the disutility from preparing to attack without eventually attacking, from attacking alone, and from attacking when conditions are not ripe is infinite, then for a large family of such technologies, an equilibrium in which both generals attack exists if and only if two events exist such that it is common knowledge (as defined in this paper) between \uline{the first general as the first event occurs} and \uline{the second general as the second event occurs} that the conditions will be ripe at sundown and that each of the two events occurs at least an hour before sundown.\footnote{If the above disutilities are finite, then an analogous characterization holds, replacing common knowledge with common $p$ belief, relaxed as we have relaxed common knowledge in this paper (i.e., common $p$ belief between \uline{the first general as the first event occurs} and \uline{the second general as the second event occurs}).} In \cref{game}, we characterize equilibrium behavior in this way for a rich family of coordinated-attack games; despite different games in this family exhibiting seemingly very disparate equilibria, these equilibria are all equally captured by our unified characterization thanks to our novel notion of common knowledge.

\subsection{Further Related Literature}\label{related}

The term  common knowledge was coined and explicitly defined by the philosopher David \citet{Lewis69}, and its study within economics was initiated by \citet{Aumann76,Aumann99}. \citeauthor{Aumann76}'s analysis, and indeed most early analyses of common knowledge in economics \citep[e.g.,][]{Milgrom81,MilgromS82,BrandenburgerD87,AumannB95}, are carried out in a static context, in which knowledge does not evolve over time. In such a setting, both the traditional definition of common knowledge and our definition coincide.  Using our notion of common knowledge, these analyses can be closely followed even in dynamic settings in which traditionally defined common knowledge does not arise.
Notable among these foundational economic papers on common knowledge is \citet{GeanakoplosP82}, who perform a dynamic rather than static analysis, albeit with no timing frictions (i.e., no temporal uncertainty on the timing of events). As we show, adding even slight timing frictions to their model unearths a vast separation between common knowledge under the traditional definition---which is no longer attained---and under our definition. \citet{SteinerS11} analyze a phenomenon similar to the Halpern--Moses Paradox in their study of when and how (probabilistic) common learning fails in a setting in which clocks are synchronized. We note that our notion of common knowledge does arise in all of the settings introduced in \citet{SteinerS11}, including those for which they prove that common learning (as traditionally defined) fails, so long as at least one message is sent and guaranteed to eventually be delivered (that is, if the distribution of message delays is, in their terms, ``not defective''). For discussions of the Halpern--Moses Paradox within economics, see also \citet{ArrowAK87,Aumann89,MorrisS97,Morris14}. Surveys on common knowledge in economics and game theory include \citet{BinmoreB88,Brandenburger92,Geanakoplos94,DekelG97}.

Within computer science, the importance of common knowledge to AI was shown by \citet{McCarthy78}, and to distributed computing by \citet{HalpernM90}. This latter paper, due to the nature of distributed computing, conducts its analysis in a dynamic setting, and among its results uncovered what would become known as the Halpern--Moses Problem/Paradox. See, e.g., \citet{ChandyM86,DworkM90,MosesT88,HalpernZ92,FaginHMV95,CastanedaGM22,CastanedaGM16} for uses of epistemic analysis in the study and design of (dynamic) distributed computing systems.
\citet{HalpernM90} also initiated a literature on variants of common knowledge defined for the goal of coordinating actions \citep[see, e.g.,][]{BenzviM13,GonczarowskiM13,FriedenbergH23}. This most recent paper, \citet{FriedenbergH23}, written concurrently and independently from our paper, defines a notion of ``action-stamped common belief.'' In a nutshell, it corresponds to ``$\alpha$, if and when it acts, believes that $\beta$, if and when it acts, believes that $\alpha$, if and when it acts, believes that\ldots.'' While related to our notion (in particular, their notion of reachability and ours are closely related from a mathematical perspective), there are qualitative differences; for example, their notion does hold in the setting of the email game of \citet{Rubinstein89} while, as already noted, ours does not. They define and apply their notion without relation to the Halpern--Moses Paradox, but rather to capture an epistemic notion underlying joint action by distributed agents that follow a shared plan (closest perhaps to our ``coordinated attack'' example above, but not in a game setting).

We define common knowledge of~$\phi$ in terms of  local events, i.e., ``$\alpha$, when her local event~$\psi_{\alpha}$ occurs, knows that $\beta$'s local event~$\psi_{\beta}$ occurs at some point and that whenever it does, $\beta$ knows that $\alpha$'s local event~$\psi_{\alpha}$ occurs at some point and that whenever it does, $\alpha$ knows \ldots\ that $\phi$.''
Our analysis naturally applies also to settings in which players are modeled using automata, and local events are captured by sets of local states of the respective automaton. Our definition of common knowledge of~$\phi$ then becomes ``$\alpha$, when in one of the states in~$\psi_{\alpha}$, knows that $\beta$ reaches one of the states in~$\psi_{\beta}$ at some point and that whenever it does, it knows that $\alpha$ reaches one of the states in~$\psi_{\alpha}$ at some point and that whenever it does, it knows \ldots\ that $\phi$.'' Modeling players as automata is customary in computer science, and to our knowledge was first employed in economics by \citet{Neyman85} and by \citet{Rubinstein86}.

\section{Preliminaries}\label{prelims}

\subsection{Model}\label{model}

\paragraph{Model Primitives}
We generally follow the notation of \citet[Section~7]{Geanakoplos94}.
We consider a finite set of \emph{players} $\allplayers$ and a set $\histories$ of objects that we call \emph{histories} (these are called \emph{states of nature} by \citealp{Geanakoplos94}), each of which is interpreted as an abstract representation of a possible complete history (for all time), or chronology, of the world. We model \emph{Time} as the natural\footnote{As in mathematical logic, we consider the natural numbers to contain zero.} numbers, $\Time\eqdef\NN$, and write $\points\eqdef\histories\times\Time$ for the set of all history-time pairs---to which we refer as \emph{points}---each uniquely identifying a specific time in a specific history. Each player has a \emph{knowledge partition} $\partition_i$ over $\points$; we write $(\history,t)\sim_i (\history',t')$, denoting that $(\history,t)$ and $(\history',t')$ are \emph{indistinguishable} in the eyes of player $i$ (when player $i$ is at either of these points), if $(\history,t),(\history',t')$ are in the same \emph{ken} (knowledge partition cell) of player $i$'s partition, i.e., if $(\history,t),(\history',t')\in\kappa$ for the same ken $\kappa\in\partition_i$. A major departure from the setting of \citet{Geanakoplos94} is that we \emph{do not} assume common knowledge of the global time (what \citealp{Geanakoplos94} refers to as all players being ``aware of the time''). That is, it is possible that $(\history,t)\sim_i (\history',t')$ while $t\ne t'$. As we will see, dispensing with this assumption has far-reaching implications.\footnote{One immediate implication is that unlike in the setting of \citet{Geanakoplos94}, a player $i$'s knowledge partition $\partition_i$ of $\points$ \emph{cannot} in general be represented as a sequence of ``time slice'' partitions $(P_{it})_{t\in\Time}$, each partitioning~$\histories$ ``at absolute time $t$.''
}
The set of players $\allplayers$, set of histories $\histories$ (and hence set of points $\points$), and the partitions $(\partition_i)_{i\in\allplayers}$ are implicit in the following definitions.

\paragraph{Events, Knowledge and Local Events} An \emph{event} is a subset of $\points$.\footnote{Another implication of not assuming common knowledge of the global time is that unlike in the analysis of \citet{Geanakoplos94}, many events of interest in our analysis are \emph{not} of the form $E\times\{t\}$ for some $E\subseteq\histories$ and $t\in T$. (Events of this form are called \emph{dated events} in \citealp{Geanakoplos94}.)}
For an event~$\phi$, we use $\lnot\phi$ to refer to the event ``$\phi$ does \emph{not} hold.'' Formally, we define $\lnot\phi\eqdef\points\setminus\phi$.
For an event~$\psi$ and an event $\phi$, we use $(\psi\rightarrow\phi)$ to refer to the event ``if~$\psi$ holds, \emph{then} so does also~$\phi$.'' Formally, we define $(\psi\rightarrow\phi)\eqdef
\lnot
(\psi\setminus\phi)$.

For a player $i\in\allplayers$ and an event $\phi\subseteq\points$, we denote by $K_i\phi$ the event ``$i$ \emph{knows} that $\phi$ is occurring,'' formally defined as
\[K_i\phi~\eqdef~\bigcup\,\{\kappa\in\partition_i \mid \kappa\subseteq\phi\}.\]
I.e., $(\history,t)\in K_i\phi$ if the ken of $(\history,t)$ in player $i$'s partition $\partition_i$ is wholly contained in~$\phi$. Since $K_i\phi$ is itself an event, this definition of knowledge makes nested knowledge events such as $K_j K_i\phi$, etc., well defined. Note that by definition, $K_i\phi\subseteq\phi$, that is, whenever $i$ knows that $\phi$ is occurring, $\phi$ is in fact indeed occurring. This is often referred to as the ``Knowledge Property'' (as it distinguishes knowledge from belief) or the ``Truth Axiom'' (only true things are known). An event $\phi$ is said to be \emph{local} to player $i\in\allplayers$, or \emph{$i$-local}, if the reverse implication holds as well, i.e., if $\phi=K_i\phi$ (\citealp{Geanakoplos94} calls such events ``self-evident 
to $i$''). That is, if additionally, whenever $\phi$ occurs, $i$ knows that $\phi$ is in fact occurring. We note that an event $\phi$ is $i$-local if and only if it is a union of kens of player $i$'s partition~$\partition_i$. Consequently, $\phi$ is $i$-local if and only if $\lnot\phi$ is $i$-local.

\paragraph{Traditionally Defined, ``Instantaneous'' Common Knowledge}
Let $\players\subseteq\allplayers$ be a set of players and let $\phi$ be an event. We define $E_\players\phi$, the event ``\emph{everyone} in $\players$ knows that $\phi$ is occurring,'' as $\bigcap_{i\in\players}K_i\phi$. For each $m\in\NN$ we define $E^m_\players\phi$ to be $m$-fold composition of $E_\players$ applied to $\phi$, that is, the event ``everyone in $\players$ knows that \ldots\ ($m$ times in total) \ldots\ everyone in $\players$ knows that $\phi$ is occurring.'' We then define $C_\players\phi$, the event ``it is \emph{common knowledge} (traditionally defined) among the players in~$\players$ that $\phi$ is occurring,'' as 
\[C_\players\phi\,~\eqdef~\, \bigcap_{m=1}^\infty E^m_\players\phi.\] We note that $C_\players\phi$ is $i$-local for every $i\in\players$, i.e., $K_i C_\players\phi=C_\players\phi$ for every $i\in\players$ and every event $\phi$.\footnote{\citet{MondererS89} call an event that is local to every player ``evident knowledge,'' while  \citet{Geanakoplos94} calls such an event a ``public event.''} In fact, it is well known that $C_\players\phi$ is the largest event contained in $\phi$ that is $i$-local for every $i\in\players$.

\subsection{The Halpern--Moses Paradox}\label{hm-paradox}

Let $\phi$ be an event that in some history $\history$ is commonly known at some time $t'>0$ but not at time $t'\!-\!1$. That is, $(\history,t')\in C_\players\phi$, while $(\history,t'\!-\!1)\notin C_\players\phi$. Since $C_\players\phi$ is $i$-local for all $i\in\players$, this means that for each $i\in\players$ we have $(\history,t')\in K_i C_\players\phi$ but $(\history,t'\!-\!1)\notin K_i C_\players\phi$. That is, 
for new common knowledge to arise, the knowledge \emph{that common knowledge has arisen} must be obtained simultaneously by each and every player $i\in\players$. Since in many natural settings, timing frictions prevent such simultaneity, the conclusion is that in such settings, new common knowledge cannot arise. That is, a fact that is not common knowledge at time $t\!=\!0$ never becomes common knowledge. Note that in real life, true simultaneity of perception, as well as certainty in that simultaneity, is never guaranteed. Taking the granularity of time to be sufficiently fine to reveal this, e.g., on the order of milliseconds, it follows that \emph{in reality, new common knowledge never arises.} This is the essence of the Halpern--Moses Paradox.

As a simple example of this paradox within our model, consider the following example adapted from \citet{HalpernM90}.\footnote{A similar example was utilized by \citet{SteinerS11} in their study of when and how (probabilistic) common learning fails in a setting in which there is common knowledge of the global time.} There are two players $\alpha$ and~$\beta$, whose watches might not be perfectly synchronized. Imagine player $\alpha$ sending a message to player $\beta$ at some time $t$ (this ``actual true time'' $t$ might not be known to any of the players), and that the message is guaranteed to arrive either one or two time units later.
In addition, suppose for simplicity that the two players have no interaction except for this message.
Consider two histories. In the first history $\history_1$, both players have accurate clocks, $\alpha$ sends the message at time~$t$, and the message is delivered at time~$t\!+\!1$. In the second history $\history_2$, player $\alpha$'s clock accurately shows the true time however $\beta$'s clock runs one unit slower than $\alpha$'s. In this history, $\alpha$ sends the same message at time~$t$, and the message is delivered at time~$t\!+\!2$, which $\beta$ sees as $t\!+\!1$ on her clock. 
Assume by way of contradiction that the content of the message becomes common knowledge at some time $t'>t$ in $\history_1$. In particular, in $\history_1$ both players know at time~$t'$, but not at time~$t'\!-\!1$, that the content is common knowledge. Since $\alpha$'s view is the same in both histories (that is, she finds both histories indistinguishable), she also knows at time~$t'$ in $\history_2$ that the content is common knowledge. However, $\beta$'s view is shifted by one time unit between the two histories, and for her $\history_2$ at (true) time $t'$ is indistinguishable from $\history_1$ at (true) time $t'\!-\!1$. Therefore, $\beta$ does not know at time~$t'$ in $\history_2$ that the content is common knowledge---a contradiction, since at time~$t'$ in $\history_2$ it must be that either both players or no players know that the content is common knowledge. 

This example is just the tip of an iceberg. Indeed, consider any back-and-forth correspondence initiated by a message from player $\alpha$ at a time unknown to player~$\beta$, and assume that each message (from $\alpha$ to $\beta$ or vice versa) is guaranteed to arrive either one or two time units after it is sent. Consider two histories that are identical except that in one, all messages sent by~$\alpha$ take one time unit to arrive and all messages sent by $\beta$ take two time units to arrive, and in the other, all messages sent by $\alpha$ take two time units to arrive and all messages sent by~$\beta$ take one time unit to arrive. Observe that both histories are indistinguishable in the eyes of player $\alpha$, and indistinguishable---however with a shift of one time unit (as in the simpler example above)---in the eyes of player $\beta$. By virtually the same argument as above, in neither of these two histories can any fact that is not commonly known at the time of the sending of the first message ever become commonly known at any later time. 

More generally, we use the following setting as a fairly general running example throughout this paper. This setting is inspired by the \emph{DCMAK} (\emph{Dynamically Consistent Model of Action and Knowledge}) setting of \citet[Section~7]{Geanakoplos94}, and can be thought of as a two-player version thereof, generalized to the model introduced in \cref{model} above, and with added timing frictions.
In this setting, which we call a Bilateral DCMAK with Timing Frictions (or BDTF, for short) and which we define formally in \cref{dcmak} below, there are two players: $\alpha$ and $\beta$. Each player $i$
does not know the objective, absolute time, but rather only knows her \emph{subjective time} $t-z^i$, where $z^i$ is a history-dependent number that we call the \emph{birth date} of player $i$.
Each player at each step sends a signal to the other player that is specified by the sender's \emph{signalling function}; this signal is received with a delay of $d^i$ where $i$ is the receiver. (I.e., a signal sent at absolute time $t$ is received by~$i$ at absolute time $t+d^i$; the delay varies across histories but is fixed throughout a given history.) Each player's initial knowledge is described by some initial partition and it evolves dynamically based on the signals that she receives.
Finally, we model each player $i$ as conscious (i.e., aware of her subjective time and able to send and receive signals) only starting at her subjective time~$0$ (i.e., from time~$t=z^i$ onward).\footnote{If any player $i$ were always conscious starting precisely at $t=0$, then this player could figure out her birth date $z^i$ by checking the subjective time at her first instant of consciousness, and the players could together figure out the difference $z^i-z^j$.
In the settings that we analyze, in contrast, this difference never becomes known.}
This last modeling feature could have been dispensed with had time been modeled as the integers rather than as the natural numbers.

\begin{definition}[Bilateral DCMAK with Timing Frictions (BDTF)]\label{dcmak}
In a \emph{Bilateral Dynamically Consistent Model of Action and Knowledge with Timing Frictions (henceforth, BDTF)}, there are two players $\allplayers\eqdef\{\alpha,\beta\}$. There is a set $O$ called the set of \emph{initial conditions}, and the set of histories is
$\histories\eqdef O\times\NN^2\times\bigl(\NN\setminus\{0\}\bigr)^2$.
For a history $\history\in\histories$, we write $\history=(o_\history,z^\alpha_\history,z^\beta_\history,d^\alpha_\history,d^\beta_\history)$.
For $i\in\allplayers$, we call $z^i_\history$ and $d^i_\history$ player~$i$'s \emph{birt hdate} and  \emph{delay} (in~$\history$), respectively. 
Each player $i\in\allplayers$ has a \emph{signal space} $S_i$ and a \emph{signalling function} $f_i:\points\rightarrow S_i$ that is measurable with respect to the partition~$\partition_i$ and which satisfies for every $\history\in\histories$ and $t<z^i_\history$ that $f_i(\history,t)=\emptyset$ (i.e., $i$ sends no signals before $i$ is conscious). Slightly abusing notation, we also write $f_i(\history,t)=\emptyset$ for every $\history\in\histories$ and $t<0$. For each player $i\in\allplayers$ there is a partition $P^0_i$ over $O$, called $i$'s \emph{initial partition}. For every player $i\in\allplayers$ and every pair  of points $(\history,t),(\history',t')\in\points$, the partition $\partition_i$ of player $i$ over $\points$ satisfies that $(\history,t)\sim_i(\history',t')$ if and only if either of the following holds:
\begin{itemize}
\item 
$t-z^i_\history<0$, $t'-z^i_{\history'}<0$, and $o_\history$ and $o_{\history'}$ are in the same partition cell of the partition~$P^0_i$ of~$O$. (Player~$i$ cannot distinguish between points at which she is not conscious except based on the initial information.)
\item
$t-z^i_\history=t'-z^i_{\history'}=0$ and both $o_\history$ and $o_{\history'}$ are in the same partition cell of the partition~$P^0_i$ of~$O$.
\item 
$t-z^i_\history=t'-z^i_{\history'}>0$ and both of the following hold:
\begin{itemize}
    \item $(\history,t-1)\sim_i(\history',t'-1)$ and
    \item $f_j(\history,t-d^i_\history)=f_j(\history',t'-d^i_{\history'})$, where $\{j\}=\allplayers\setminus\{i\}$.
\end{itemize}
\end{itemize}
\end{definition}

Arguments similar to the one preceding the introduction of BDTF in \Cref{hm-paradox} give rise to the following \lcnamecref{no-ck}.

\begin{theorem}[No New Common Knowledge with Timing Frictions {\citep[see also][]{HalpernM90,SteinerS11}}]\label{no-ck}
Consider any BDTF. Let $\history$ be a history such that $d^\alpha_\history+d^\beta_\history>2$. If a fact is not common knowledge (as traditionally defined) between $\alpha$ and $\beta$ at time $0$ in $\history$, then it is never common knowledge between them at any later time in $\history$.
\end{theorem}

We note that the condition $d^\alpha_\history+d^\beta_\history>2$ in \cref{no-ck} is a technical condition that has to do with the discrete modeling of time. The \emph{round-trip delay} $d^\alpha_\history+d^\beta_\history$ might become known by both players during the history $\history$ (see the proof of \cref{get-ck} below). The technical condition in \cref{no-ck} ensures that the round-trip delay does not uniquely identify the two individual delays (since $1$ is the minimum possible delay). Indeed, if either of the individual delays becomes known, new common knowledge can arise. This technical condition could have been dispensed with had time been modeled in such a way that there were no smallest possible delay.

\begin{proof}[Proof of \cref{no-ck}]
Let $\history$ be a history such that $d^\alpha_\history+d^\beta_\history>2$. To prove the claim, it suffices to show that, for all times~$t$ and events~$\phi$, if $(\history,t)\in\lnot C_\allplayers\phi$ then $(\history,t\!+\!1)\in\lnot C_\allplayers\phi$. 
Fix $t$ and~$\phi$, and suppose that $(\history,t)\in\lnot C_\allplayers\phi$. 
Note that since $C_\allplayers\phi$ is $i$-local for every $i\in\allplayers$, so is $\lnot C_\allplayers\phi$.

Since $d^\alpha_\history+d^\beta_\history>2$, there exists $j\in\allplayers$ such that $d^j_\history>1$. Let $j$ be such a player and let $i$ be the other player. Consider the history $\history'$ defined such that $(o_{\history'},z^i_{\history'},z^j_{\history'},d^i_{\history'},d^j_{\history'})=(o_\history,z^i_\history\!+\!1,z^j_\history,d^i_\history\!+\!1,d^j_\history\!-\!1)$. By definition of $\history'$, we have that $(\history,t)\sim_i(\history',t\!+\!1)$ and that $(\history,t\!+\!1)\sim_j(\history',t\!+\!1)$. Since $(\history,t)\in\lnot C_\allplayers\phi$ and since $\lnot C_\allplayers\phi$ is $i$-local, by the former we have that $(\history',t\!+\!1)\in\lnot C_\allplayers\phi$ and since $\lnot C_\allplayers\phi$ is $j$-local, by the latter we then have that $(\history,t\!+\!1)\in\lnot C_\allplayers\phi$, as required.
\end{proof}

We defined BDTF as a reasonably general model so that the positive results that we prove about it in later sections are also meaningful. Since \cref{no-ck} is a negative result, we note that it holds even in much more restrictive models (in which much more is common knowledge to begin with). Specifically, even if, for some integer $D>2$ we were to restrict any BDTF only to ``single dimensional'' timing frictions satisfying $z^\alpha_\history=0$, $z^\beta_\history=d^\beta_\history$,  and $d^\alpha_\history=D-d^\beta_\history$ (in which case, for example, the value of~$D$ and the correctness of all of these equations would be common knowledge to begin with; furthermore, in this case if $O$ is finite, so is $\histories$), \cref{no-ck} would still hold.\footnote{The proof is similar, albeit if the history $\history'$ constructed in the proof of \cref{no-ck} does not satisfy the restriction for any choice of $i,j$, the proof instead turns to the history $\history''$ defined using $(o_\history,z^i_\history,z^j_\history\!-\!1,d^i_\history\!+\!1,d^j_\history\!-\!1)$ for suitable $i,j$ and proceeds by noting that that $(\history,t)\sim_i(\history'',t)$ and $(\history'',t)\sim_j(\history,t\!+\!1)$.}

As a special case, \cref{no-ck} precludes a guarantee of attaining common knowledge of posteriors in the setting of \citet{GeanakoplosP82} (in which signals containing updated posteriors are sent back and forth) if even slight timing frictions are introduced (and even if much more is transmitted at any round than merely updated posteriors).

\section{Overcoming the Tyranny of the Clock}\label{attaining}

Having reviewed the Halpern--Moses Paradox, in this section we define our notion of common knowledge and derive necessary and sufficient conditions for it to arise. In particular, we show that it arises even in settings that exhibit the Halpern--Moses Paradox, such as the \citet{GeanakoplosP82} setting with added timing frictions.

\subsection{Common Knowledge, Redefined}

For an event $\phi$, we define $\histories(\phi)\eqdef\bigl\{\history\in\histories~\big|~(\{\history\}\times\Time)\cap\phi\ne\emptyset\bigr\}$. I.e., $\histories(\phi)$ is the set of histories during which $\phi$ occurs (at least once).
For an event $\phi$, we use $\always\phi$ to refer to the event ``\emph{throughout} this entire history, $\phi$ holds.''
Formally, we define $\always\phi\eqdef\bigcup\bigl\{\{\history\}\times\Time~\big|~\{\history\}\times\Time\subseteq\phi\bigr\}$. 
A second temporal operator that will serve us is the dual of $\always$, denoted by~$\sometime$, where $\sometime\phi$ denotes the time-invariant event ``at \emph{some time} in the current history, $\phi$ holds.'' Formally, $\sometime\phi\eqdef \histories(\phi)\times\Time$.
We say that an event $\phi$ is \emph{time-invariant} if $\phi=\sometime\phi$ (equivalently, if $\phi=\always\phi$). Note that $\always\phi$ and $\sometime\phi$ are time-invariant for every event $\phi$.

For a player~$i$, an $i$-local event $\psi_i$, and an event $\phi$,
we denote by $K_{i\subat\psi_i}\phi$ the time-invariant event ``$\psi_i$ holds at some time during this history, and whenever it does, player~$i$ knows that $\phi$.'' Formally,
$K_{i\subat\psi_i}\phi\,\eqdef\,
\sometime\psi_i\cap\always(\psi_i\rightarrow K_i\phi$). 
In what follows, it will sometimes be useful to restrict attention to events $\psi_i$ that are furthermore \emph{singular}, i.e., occur at most once throughout any given history. 
In this special case, $
\always(\psi_i\rightarrow K_i\phi)$ means ``\emph{at the time at which} $\psi_i$ holds in the current history, $K_i\phi$ holds as well.''

Let $\players\subseteq\allplayers$ be a set of players. An \emph{$\players$-profile} is a tuple $\bar{\psi}=(\psi_i)_{i\in\players}$ such that $\psi_i$ is an $i$-local event for every $i\in\players$. Let $\bar{\psi}$ be such a profile and let~$\phi$ be an event. We define $E_{\players\subat\bar{\psi}}\phi\eqdef\bigcap_{i\in\players}K_{i\subat\psi_i}\phi$, meaning ``each $\psi_i$ holds at some time during this history, and whenever one of them does, the player $i$ in question knows that $\phi$.'' For every $m\in\NN$ we use $E^m_{\players\subat\bar{\psi}}$ to denote the $m$-fold composition of $E_{\players\subat\bar{\psi}}$.

\begin{definition}[Common Knowledge]\label{ck}
Let $\players$ be a set of players, let $\bar{\psi}$ be an $\players$-profile, and let~$\phi$ be an event.
We define $C_{\players\subat\bar{\psi}}\phi$ as the time-invariant event $\bigcap_{m=1}^{\infty}E^m_{\players\subat\bar{\psi}}\phi$.
\end{definition}

If the time-invariant event $C_{\players\subat\bar{\psi}}\phi$ holds at (every point throughout) a certain history, then, denoting $\players=\{i_1,\ldots,i_n\}$, we say that in that history the event~$\phi$ is \emph{common knowledge} between $i_1@\psi_{i_1}$ (that is, $i_1$ as $\psi_{i_1}$ holds), $i_2@\psi_{i_2}$ (that is, $i_2$ as $\psi_{i_2}$ holds), \ldots, and $i_n@\psi_{i_n}$ (that is, $i_n$ as $\psi_{i_n}$ holds). 

We emphasize that $(\history,t)\in C_{\players\subat\bar{\psi}}\phi$ does not mean that $(\history,t)\in\phi$. (If this were the case, then since $C_{\players\subat\bar{\psi}}\phi$ is time-invariant, the event $\phi$ would need to hold throughout the entire history~$\history$.) Rather, $(\history,t)\in C_{\players\subat\bar{\psi}}\phi$ means that for every $i\in\players$, the $i$-local event $\psi_{i}$ holds at some point in $\history$ (possibly at a time other than $t$), and whenever it does, the event $\phi$ also holds, and $i$ knows that $\phi$ holds, and $i$ knows that whenever any~$\psi_j$ holds, $\phi$ also holds, and $j$ knows that $\phi$ holds, etc. It follows that, for each $i\in\players$, the precise event in which $i$ participates in this joint state of common knowledge is $C_{\players\subat\bar{\psi}}^i\phi\eqdef\psi_i\cap C_{\players\subat\bar{\psi}}\phi$. 

\begin{lemma}\label{cki}
Let $\players$ be a set of players, let $\bar{\psi}$ be an $\players$-profile, and let~$\phi$ be an event. For every player $i\in\players$, we have that (1)~$\histories(C_{\players\subat\bar{\psi}}^i\phi)=\histories(C_{\players\subat\bar{\psi}}\phi)$, and (2)~$C_{\players\subat\bar{\psi}}^i\phi\subseteq\phi$.
\end{lemma}

\begin{proof}
For the first part, observe that\[\histories(C_{\players\subat\bar{\psi}}^i\phi)=\histories(\psi_i\cap C_{\players\subat\bar{\psi}}\phi)=\histories(\psi_i)\cap \histories(C_{\players\subat\bar{\psi}}\phi)=\histories(C_{\players\subat\bar{\psi}}\phi),\]
where the second equality follows from the fact that $C_{\players\subat\bar{\psi}}\phi$ is time-invariant and the last equality is since $C_{\players\subat\bar{\psi}}\phi\subseteq K_{i\subat\psi_i}\phi\subseteq\sometime\psi_i$. For the second part, we have that
\[C_{\players\subat\bar{\psi}}^i\phi=\psi_i\cap C_{\players\subat\bar{\psi}}\phi\subseteq \psi_i\cap K_{i\subat\psi_i}\phi\subseteq \psi_i\cap\always(\psi_i\rightarrow K_i\phi)\subseteq K_i\phi\subseteq\phi.\qedhere
\]
\end{proof}

A question that we are sometimes asked is ``When does $\phi$ become common knowledge under your definition?'' There are two perspectives one might take here: One, which is easier to see from \cref{ck}, is that common knowledge of $\phi$ between $i_1@\psi_{i_1},\ldots,i_n@\psi_{i_n}$ is a time-invariant event that either holds in a given history---i.e., whenever any $\psi_{i_j}$ holds throughout this history, $i_j$ knows the relevant facts---or does not hold in that history. Another perspective one might take here is that common knowledge arises at different times for different players, i.e., common knowledge holds for each player $i_j$ whenever $C_{\players\subat\bar{\psi}}^{i_j}\phi$ holds. Recall that this happens in every history in which the time-invariant event $C_{\players\subat\bar{\psi}}\phi$ holds, at each instant at which $\psi_{i_j}$ holds. The fact that these events, and the times at which they hold in any given history, might differ across the various players is due to the asymmetry and non-simultaneity that our notion of common knowledge admits (which allow it to be attainable even in settings that exhibit the Halpern--Moses Paradox). This is in contrast with traditionally defined common knowledge, which arises among all players $i_j$ simultaneously and hence for each of them the event in which she participates is precisely the same: $C_\players\phi=K_{i_j}C_\players\phi$.\footnote{Instead of defining $C_{\players\subat\bar{\psi}}\phi$ and deriving the individualized events ${C_{\players\subat\bar{\psi}}^{i_1}\phi=\psi_{i_1}\cap C_{\players\subat\bar{\psi}}\phi},$ $\,\ldots\,,C_{\players\subat\bar{\psi}}^{i_n}\phi=\psi_{i_n}\cap C_{\players\subat\bar{\psi}}\phi$ from it, we alternatively could have directly defined these individualized events as our building blocks. One way to do that is in the spirit of \cref{ck-fixed-point} below: The tuple $(C_{\players\subat\bar{\psi}}^{i_1}\phi,\ldots,C_{\players\subat\bar{\psi}}^{i_n}\phi)$ can be defined as the greatest fixed point of the (vectorial) function $(\chi_1,\ldots,\chi_n)\mapsto\bigl(K_{i_1}(\psi_{i_1}\cap\phi\cap\bigcap_{i\ne i_1}(\sometime\psi_i\cap\always(\psi_i\rightarrow\chi_i))),\ldots,K_{i_n}(\psi_{i_n}\cap\phi\cap\bigcap_{i\ne i_n}(\sometime\psi_i\cap\always(\psi_i\rightarrow\chi_i)))\bigr)$. Equivalently, we could have defined each of these individualized events as a distinct infinite intersection of events in the spirit of \cref{ck}. While for some asymmetric variants of common knowledge \citep{GonczarowskiM13} we do not know of a way to avoid using one of these approaches, in our setting we are able to avoid using them, resulting in what we not only view as a technically simple definition, but also as a conceptually better one since using it, the phrase ``common knowledge holds/is attained'' also technically, and not only conceptually, refers to a single event rather than to multiple events. This also simplifies and clarifies the statement of some of our results, such as \cref{induction-rule} below.}

As we now show, much as traditionally defined common knowledge ($C_\players\phi$) is local to each player (so each player knows when traditionally defined common knowledge holds), each of the individualized events $C_{\players\subat\bar{\psi}}^i\phi=\psi_{i}\cap C_{\players\subat\bar{\psi}}\phi$ that we just discussed is local for its respective player~$i$ (so each player knows when her own individualized event holds). 

\begin{lemma}[Locality of Common Knowledge]\label{ck-local}
Let $\players$ be a set of players, let $\bar{\psi}$ be an $\players$-profile, and let~$\phi$ be an event.
For every player $i\in\players$, the event $C_{\players\subat\bar{\psi}}^i\phi=\psi_i\cap C_{\players\subat\bar{\psi}}\phi$ is $i$-local.
\end{lemma}

\begin{proof}
It suffices to show that $\psi_i\cap C_{\players\subat\bar{\psi}}\phi\subseteq K_i(\psi_i\cap C_{\players\subat\bar{\psi}}\phi)$. Indeed,
\begin{multline*}
\psi_i\cap C_{\players\subat\bar{\psi}}\phi\subseteq
\psi_i\cap K_{i\subat\psi_i}C_{\players\subat\bar{\psi}}\phi\subseteq
\psi_i\cap \always(\psi_i\rightarrow K_i C_{\players\subat\bar{\psi}}\phi)\subseteq\\*
\subseteq\psi_i\cap K_i C_{\players\subat\bar{\psi}}\phi=K_i(\psi_i\cap C_{\players\subat\bar{\psi}}\phi),
\end{multline*}
where the first inclusion is by definition of $C_{\players\subat\bar{\psi}}\phi$ and since $K_{i\subat\psi_i}(\cdot)$ commutes with intersection, the second inclusion is by definition of $K_{i\subat\psi_i}$, and the equality follows since $\psi_i$ is $i$-local.
\end{proof}

By \cref{ck-local}, if $C_{\players\subat\bar{\psi}}\phi$ holds in a history~$\history$, then when $\psi_i$ holds in $\history$, player~$i$ (in addition to knowing $\phi$ as discussed above) \emph{knows} that $C_{\players\subat\bar{\psi}}\phi$ holds.
This leads to an alternative, equivalent definition of common knowledge as a fixed point, which will also be useful in our analysis later in this paper. This definition is inspired by an analogous definition of (traditionally defined) common knowledge, which dates back explicitly to \citet{Harman77} (see also \citealp{Barwise88}), and implicitly to \citet{Aumann76}.\footnote{This definition formulates $C_I\phi$ as the greatest fixed point of the function $\chi\mapsto E_\players(\phi\cap\chi)$.}

\begin{lemma}[Common Knowledge as a Fixed Point]\label{ck-fixed-point}
Let $\players$ be a set of players, let $\bar{\psi}$ be an $\players$-profile, and let~$\phi$ be an event.
\begin{itemize}
\item
The function $\chi\mapsto E_{\players\subat\bar{\psi}}(\phi\cap\chi)$ has a greatest fixed point.
\item
$C_{\players\subat\bar{\psi}}\phi$ is the greatest fixed point of the function $\chi\mapsto E_{\players\subat\bar{\psi}}(\phi\cap\chi)$.
\end{itemize}
\end{lemma}

\begin{proof}
Part 1: Since the operators $\rightarrow$ (in its right operand), $K_i$, $\always$, $\cap$, and $\sometime$ are all monotone, the function $\chi\mapsto E_{\players\subat\bar{\psi}}(\phi\cap\chi)$ is monotone. Therefore, by Tarski's fixed-point theorem, it indeed has a greatest fixed point.

Part 2: The same function $f(\chi)=E_{\players\subat\bar{\psi}}(\phi\cap\chi)$ commutes with intersection, and is thus downward-continuous. Therefore, by Kleene's fixed-point theorem, its greatest fixed point, $C_{\players\subat\bar{\psi}}\phi$, equals $\bigcap_{m=1}^{\infty}f^m(\points)=\bigcap_{m=1}^{\infty}E^m_{\players\subat\bar{\psi}}\phi$, as claimed.
\end{proof}

\subsection{Induction Rule}\label{induction}

As already noted by \citet{ClarkM81}, ascertaining that common knowledge holds is quite an arduous task that requires ensuring that infinitely many events hold (or, we might add---using \cref{ck-fixed-point}---that an implicitly defined event holds). To make this task more practical in analyses, we define an \emph{Induction Rule} for our variant of common knowledge.

\begin{theorem}[Induction Rule]\label{induction-rule}\leavevmode
Let $\players$ be a set of players, let $\bar{\psi}$ be an $\players$-profile, and let~$\phi$ be an event. If $\phi\subseteq E_{\players\subat\bar{\psi}}\phi$, then $\phi\subseteq C_{\players\subat\bar{\psi}}\phi$.
\end{theorem}

In the Induction Rule, set inclusion (``$\subseteq$'') should be interpreted as implication between events. E.g., the condition $\phi\subseteq E_{\players\subat\bar{\psi}}\phi$ should be interpreted as saying ``the event $\phi$ implies the event $E_{\players\subat\bar{\psi}}\phi$,'' (this is equivalent to the more cumbersome statement $(\phi\rightarrow E_{\players\subat\bar{\psi}}\phi)=\points$).
\cref{induction-rule} is implied by the following technical lemma.

\begin{lemma}\label{induction-lemma}
Let $\players$ be a set of players, let $\bar{\psi}$ be an $\players$-profile, and let $\phi$ be an event. If an event $\xi$ satisfies $\xi\subseteq E_{\players\subat\bar{\psi}}(\phi\cap\xi)$, then $\xi\subseteq C_{\players\subat\bar{\psi}}\phi$.
\end{lemma}

\begin{proof}
The function $E_{\players\subat\bar{\psi}}(\phi\cap\cdot)$ is monotone, and therefore, by Tarski's fixed point theorem, its greatest fixed point, which by \cref{ck-fixed-point} is $C_{\players\subat\bar{\psi}}\phi$, equals $\bigcup\bigl\{\chi\mid\chi\subseteq E_{\players\subat\bar{\psi}}(\phi\cap\chi)\bigr\}$. Therefore, since $\xi\subseteq E_{\players\subat\bar{\psi}}(\phi\cap\xi)$, we have that $\xi\subseteq C_{\players\subat\bar{\psi}}\phi$, as claimed.
\end{proof}

\begin{proof}[Proof of \cref{induction-rule}]
The \lcnamecref{induction-rule} follows from \cref{induction-lemma}, taking $\xi=\phi$.
\end{proof}

The following useful lemma allows us to use the Induction Rule (\cref{induction-rule}) to prove that $\sigma\subseteq C_{\players\subat\bar{\psi}}\phi$ by showing that $\sigma\subseteq E_{\players\subat\bar{\psi}}(\sigma)$, for events $\sigma\subset\phi$ (i.e., sufficient conditions for $\phi$) rather than only for $\sigma=\phi$.

\begin{lemma}[Monotonicity of $C_{\players\subat\bar{\psi}}$]\label{monotonicity}
Let $\players$ be a set of players, let $\bar{\psi}$ be an $\players$-profile, and let $\sigma$ and $\phi$ be events. If $\sigma\subseteq\phi$, then $C_{\players\subat\bar{\psi}}\sigma\subseteq C_{\players\subat\bar{\psi}}\phi$.
\end{lemma}

\begin{proof}
Let $\sigma\subseteq\phi$. For every $\xi$, we have that $E_{\players\subat\bar{\psi}}(\sigma\cap\xi)\subseteq E_{\players\subat\bar{\psi}}(\phi\cap\xi)$. In particular, taking $\xi=C_{\players\subat\bar{\psi}}\sigma$ we have by \cref{ck-fixed-point} that $\xi=E_{\players\subat\bar{\psi}}(\sigma\cap\xi)\subseteq E_{\players\subat\bar{\psi}}(\phi\cap\xi)$. Hence, by \cref{induction-lemma}, $C_{\players\subat\bar{\psi}}\sigma=\xi\subseteq C_{\players\subat\bar{\psi}}\phi$.
\end{proof}

The Induction Rule can be used to directly ascertain the emergence of (our notion of) common knowledge in many settings. E.g., consider the example from \cref{hm-paradox} in which $\alpha$ sends a single message to $\beta$. Let $\psi_\alpha$ be the $\alpha$-local event ``a message with content~$c$ is sent by~$\alpha$ to $\beta$,'' and let $\psi_\beta$ be the $\beta$-local event ``a message with content~$c$ is received by $\beta$ from $\alpha$.'' Then, $\sometime\psi_\alpha\subseteq K_{\alpha\subat\psi_\alpha}\sometime\psi_\alpha$ (indeed, when~$\alpha$ sends her message in this example, she surely knows this) and $\sometime\psi_\alpha\subseteq K_{\beta\subat\psi_\beta}\sometime\psi_\alpha$ (indeed, the message sent by~$\alpha$ is guaranteed to be received by $\beta$ eventually, at which point $\beta$ surely knows that it had been sent at some prior point). Therefore, by the Induction Rule, $\sometime\psi_\alpha\subseteq C_{\players\subat\bar{\psi}}\sometime\psi_\alpha$. That is, in histories in which $\alpha$ sends a message with content~$c$ to~$\beta$, the fact that such a message is sent in the history is common knowledge between $\alpha$ as she sends this message and $\beta$ as she receives it. If, in addition, we have that $\sometime\psi_\alpha\subseteq\phi$ for some event $\phi$, e.g., in the generals example in the introduction taking $\phi=$``there are favorable conditions to attack at sundown'' (the assumption $\sometime\psi_\alpha\subseteq\phi$ means in this case that $\alpha$ will not send a message with content~$c$ unless there are favorable conditions to attack at sundown), then by \cref{monotonicity} we have that in histories in which $\alpha$ sends a message with content~$c$ to $\beta$, the fact that there are favorable conditions to attack at sundown is common knowledge between $\alpha$ as she sends this message and $\beta$ as she receives it.

\subsection{Co-occurrence}\label{sec:cooccurrence}

Reflecting upon our analysis of the ``attack at sundown'' example in the end of \cref{induction}, it is useful to generally ascertain the precise conditions under which the fact that a message is sent by a player $i$ to a player $j$ implies that the content of the message is common knowledge between $i$ as she sends the message and $j$ as she receives it. In this \lcnamecref{sec:cooccurrence}, we provide a necessary and sufficient condition for this implication to hold true.

The Stanford Encyclopedia of Philosophy \citep{Zalta13} highlights events being simultaneous and public as a prerequisite for common knowledge to arise. Relatedly, \citet{ClarkM81} discuss \emph{co-presence} of all players as giving rise to common knowledge of anything that is said or observed during their interaction. Our notion of common knowledge is no longer tied to (even approximate) simultaneity or to events being public, yet it is tied to what we term as the events $\psi_i$ \emph{co-occurring}: At each history either all of the $\psi_i$ events occur, possibly at different times in the history, or none ever occurs. For example, if a message sent by a player~$\alpha$ is guaranteed to eventually be received by a player~$\beta$, then the transmission by~$\alpha$ and the receipt by~$\beta$ are co-occurring events. We now formalize the notion of co-occurrence, which is a relaxation of co-presence, and show that it provides for a necessary and sufficient condition for the emergence of common knowledge of communicated information.

\begin{definition}[Co-occurrence]
Let $\players$ be a set of players, let $\bar{\psi}=(\psi_i)_{i\in\players}$ be an $\players$-profile, and let $\histories'\subseteq\histories$ be a set of histories. We say that $\bar{\psi}$ satisfies co-occurrence in~$\histories'$ if $\histories'\times\Time\subseteq(\sometime\psi_i\rightarrow\sometime\psi_j)$ for every pair of players $i,j\in\players$.
\end{definition}

Note that in particular, an $\players$-profile satisfies co-occurrence in the set $\histories$ of all histories if and only if $\sometime\psi_i=\sometime\psi_j$ (equivalently, $\histories(\psi_i)=\histories(\psi_j)$) for every pair of players $i,j\in\players$.

\begin{theorem}[Co-occurrence and Common Knowledge]\label{cooccurrence-iff}
Let $\players$ be a set of players, let $\bar{\psi}=(\psi_i)_{i\in\players}$ be an \mbox{$\players$-profile}, and let $\phi$ be an event such that $\sometime\psi_\ell\subseteq\phi$ for some $\ell\in\players$. Then $\psi_i\subseteq C_{\players\subat\bar{\psi}}\phi$ for every $i\in\players$ if and only if $\bar{\psi}$ satisfies co-occurrence in $\histories$.
\end{theorem}

Before we prove \cref{cooccurrence-iff}, we emphasize that the expression $\psi_i\subseteq C_{\players\subat\bar{\psi}}\phi$ in the statement of this \lcnamecref{cooccurrence-iff} implies that $\psi_i\subseteq \psi_i\cap C_{\players\subat\bar{\psi}}\phi=C_{\players\subat\bar{\psi}}^i\phi$, which by \cref{ck-local} in turn implies that $\psi_i\subseteq K_i(C_{\players\subat\bar{\psi}}^i\phi)\subseteq K_i(C_{\players\subat\bar{\psi}}\phi)$.
That is, whenever $\psi_i$ holds (e.g., whenever $i$ sends or receives a message whose content implies $\phi$), it is the case that $i$ knows that it is common knowledge between $i_1@\psi_1,\ldots,i_n@\psi_n$ (one of which is $i@\psi_i$, i.e., $i$ at that instant) that $\phi$ holds. (And, in particular, $i$ at that instant also knows that $\phi$ holds.)
\cref{cooccurrence-iff} is implied by the following technical lemma.

\begin{lemma}\label{characterization}
Let $\players$ be a set of players and let $\bar{\psi}=(\psi_i)_{i\in\players}$ be an $\players$-profile.
\begin{enumerate}
\item\label{characterization-if}
 If $\bar{\psi}$ satisfies co-occurrence in $\histories$, then $\psi_i\subseteq C_{\players\subat\bar{\psi}}\bigl(\bigcap_{j\in\players}\sometime\psi_j\bigr)$ for every $i\in\players$.
\item\label{characterization-onlyif}
If for some event $\phi$ it holds that $\psi_i\subseteq C_{\players\subat\bar{\psi}}\phi$ for every $i\in\players$, then $\bar{\psi}$ satisfies co-occurrence in $\histories$.
\end{enumerate}
\end{lemma}

\begin{proof}
Part 1: For every $i\in\players$, the event $\psi_i$ is $i$-local, and hence $\psi_i=K_i\psi_i$, and therefore we have $\always(\psi_i\rightarrow K_i\psi_i)=\points$. By co-occurrence, $\bigcap_{j\in\players}\sometime\psi_j=\sometime\psi_i$. Hence, $\bigcap_{j\in\players}\sometime\psi_j=\sometime\psi_i={\sometime\psi_i\cap\always(\psi_i\rightarrow K_i\psi_i)}=K_{i\subat\psi_i}\psi_i\subseteq K_{i\subat\psi_i}\sometime\psi_i=K_{i\subat\psi_i}\bigl(\bigcap_{j\in\players}\sometime\psi_j\bigr)$. Since this holds for every $i\in\players$, it follows that $\bigcap_{j\in\players}\sometime\psi_j\subseteq E_{\players\subat\bar{\psi}}\bigl(\bigcap_{j\in\players}\sometime\psi_j\bigr)$, and so by the Induction Rule (\cref{induction-rule}), we have that $\bigcap_{j\in\players}\sometime\psi_j\subseteq C_{\players\subat\bar{\psi}}\bigl(\bigcap_{j\in\players}\sometime\psi_j\bigr)$. Therefore, for every $i\in\players$ we have that $\psi_i\subseteq \sometime\psi_i=\bigcap_{j\in\players}\sometime\psi_j\subseteq C_{\players\subat\bar{\psi}}\bigl(\bigcap_{j\in\players}\sometime\psi_j\bigr)$.

Part 2:
For every $i\in\players$ we have that $\psi_i\subseteq C_{\players\subat\bar{\psi}}\phi\subseteq E_{\players\subat\bar{\psi}}\phi=\bigcap_{j\in\players}K_{j\subat\psi_j}\phi\subseteq\bigcap_{j\in\players}\sometime\psi_j$. Therefore, for every $i,j\in \players$ we have $\sometime\psi_i\subseteq\sometime\psi_j$ (and similarly $\sometime\psi_j\subseteq\sometime\psi_i$), and so $\bar{\psi}$ satisfies co-occurrence in $\histories$.
\end{proof}

\begin{proof}[Proof of \cref{cooccurrence-iff}]
The ``only if''  direction follows from \cref{characterization}(\labelcref{characterization-onlyif}). For the ``if'' direction, let $i\in\players$. By \cref{characterization}(\labelcref{characterization-if}), we have that $\psi_i\subseteq C_{\players\subat\bar{\psi}}\bigl(\bigcap_{j\in\players}\sometime\psi_j\bigr)$; by \cref{monotonicity} and since $\bigcap_{j\in\players}\sometime\psi_j\subseteq\sometime\psi_\ell\subseteq\phi$, it follows that $\psi_i\subseteq C_{\players\subat\bar{\psi}}\phi$ as required.
\end{proof}

To see why we view \cref{cooccurrence-iff} as establishing that co-occurrence is necessary and sufficient for the emergence of common knowledge of communicated information, consider how this \lcnamecref{cooccurrence-iff} can be used to considerably simplify the analysis of the ``attack at sundown'' example from the end of \cref{induction}. Let $\psi_\alpha$, $\psi_\beta$, and $\phi$ be as defined there. Notice that $\bar{\psi}=\{\psi_\alpha,\psi_\beta\}$ satisfies co-occurrence and that $\sometime\psi_\alpha\subseteq\phi$. Therefore, by \cref{cooccurrence-iff}, we immediately have that
$\psi_i\subseteq C_{\players\subat\bar{\psi}}\phi$, for every $i\in\{\alpha,\beta\}$, as was the conclusion there. Furthermore, by the same \lcnamecref{cooccurrence-iff} the co-occurrence of $\bar{\psi}$ cannot be relaxed (e.g., to probabilistic delivery of messages as in \citealp{Rubinstein89}) and still yield the same conclusion. Since co-occurrence of $\bar{\psi}$ is equivalent to guaranteed delivery of the message sent by $\alpha$, it is precisely what we would intuitively want to have as a necessary and sufficient condition for the sending of a message to imply common knowledge of its content. Under our definition of common knowledge, this is a necessary and sufficient condition not only intuitively, but also formally.

\subsection{Reachability}

As already pointed out by \citet{Aumann76}, traditionally defined common knowledge is closely related to reachability among points. 
Namely, $(\history,t)\in C_\players\phi$ if and only if $(\history',t')\in\phi$ holds for every point $(\history',t')$ such that there is a sequence of players $i_1,\ldots,i_m\in\players$  and points $(\history,t)=(\history_1,t_1),\ldots,(\history_{m+1},t_{m+1})=(\history',t')$ such that $(\history_j,t_j)\sim_{i_j}(\history_{j+1},t_{j+1})$ for every $j\in\{1,\ldots,m\}$. As we now show, 
common knowledge is related to (an analogous definition of) reachability under our definition as well, which affords a differently flavored, yet still very useful, characterization of when common knowledge holds. 

\begin{definition}[Reachability]
Let $\players$ be a set of players and let $\bar{\psi}$ be an $\players$-profile.
\begin{enumerate}
\item
The reachability graph $G_{\players,\bar{\psi}}=(\histories,\mathcal{E}_{\players,\bar{\psi}})$ is an undirected graph where \[\mathcal{E}_{\players,\bar{\psi}}\eqdef\bigl\{\{\history,\history'\}\in\histories^2~\big|~\exists i\in \players,t,t'\in T : (\history,t),(\history',t')\in\psi_i~\&~(\history,t)\sim_i(\history',t')\bigr\}.\]
\item 
We say that a point $(\history',t')\in\points$ is \emph{reachable} (with respect to $\players$ and $\bar{\psi}$) from $(\history,t)$ if (i)~$\history'$ is in the connected component of $\history$ in $G_{\players,\bar{\psi}}$, and (ii)~there exists a player $i\in\players$ s.t.\ $(\history',t')\in\psi_i$.
\end{enumerate}
\end{definition}

\begin{theorem}[Reachability and Common Knowledge]\label{reachability}
Let $\players$ be a set of players, let $\bar{\psi}$ be an $\players$-profile,
let~$\phi$ be an event, and let $(\history,t)\in\points$. 
Then $(\history,t)\in C_{\players\subat\bar{\psi}}\phi$ if and only if
(i) some $\psi_i$ holds at some point in $\history$, 
(ii) $\bar{\psi}$ satisfies co-occurrence in all histories of the connected component of $\history$ in $G_{\players,\bar{\psi}}$, and (iii)~$(\history',t')\in\phi$ for every $(\history',t')$ that is reachable from $(\history,t)$.
\end{theorem}

\begin{proof}
Let $(\history,t)\in\points$. For a player $i\in\players$, we say that a point $(\history',t')\in\points$ is \emph{$(i,1)$-reachable} from $(\history,t)$ if
there exists $t''\in\Time$ such that $(\history,t''),(\history',t')\in\psi_i$ and $(\history,t'')\sim_i(\history',t')$. 
For $m\in\NN$, we say that $(\history',t')\in\points$ is \emph{$m$-reachable} from $(\history,t)$ if (i)~$\history$ and $\history'$ are connected via a path of length $m$ in $G_{\players,\bar{\psi}}$, and (ii)~there exist $i\in\players$ s.t.\ $(\history',t')\in\psi_i$. 

For each $i\in\players$, we observe that $(\history,t)\in K_{i\subat\psi_i}\phi$ if and only if both (i)~$\psi_i$ holds at some point in $\history$, and (ii)~$(\history',t')\in\phi$ for every $(\history',t')$ that is $(i,1)$-reachable from $(\history,t)$. Therefore, $(\history,t)\in E_{\players\subat\bar{\psi}}\phi$ if and only if both 
(i)~\emph{each} $\psi_i$ holds at some point in~$\history$, 
and (ii)~$(\history',t')\in\phi$ for every $(\history',t')$ that is $1$-reachable from $(\history,t)$. By induction, therefore, for every $m\in\NN\setminus\{0\}$ we have that $(\history,t)\in E^m_{\players\subat\bar{\psi}}\phi$ if and only if both 
(i)~each~$\psi_i$ holds at some point in every history that is connected to $\history$ via a path of length at most $m\!-\!1$ in $G_{\players,\bar{\psi}}$ (including, in particular, $\history$ itself), 
and (ii)~$(\history',t')\in\phi$ for every $(\history',t')$ that is $m$-reachable from $(\history,t)$. Since $C_{\players\subat\bar{\psi}}\phi=\bigcap_{m=1}^{\infty}E^m_{\players\subat\bar{\psi}}\phi$, we therefore have that $(\history,t)\in C_{\players\subat\bar{\psi}}\phi$ if and only if (i)~each~$\psi_i$ holds at some point in every history in the connected component of $\history$ in $G_{\players,\bar{\psi}}$ and (ii)~$(\history',t')\in\phi$ for every $(\history',t')$ that is reachable from $(\history,t)$.  \cref{reachability} follows since each $\psi_i$ holds at some point in every history in the connected component of $\history$ in $G_{\players,\bar{\psi}}$ if and only if (i)~some~$\psi_i$ holds at some point in $\history$ and (ii)~$\bar{\psi}$ satisfies co-occurrence in all histories in the connected component of $\history$ in $G_{\players,\bar{\psi}}$.
\end{proof}

\subsection{Getting to Common Knowledge despite Timing Frictions}

We conclude this \lcnamecref{attaining} by showing that the Halpern--Moses Paradox does not apply to common knowledge as we have defined it in the way that it does apply to traditionally defined common knowledge. Recall that  \cref{no-ck} establishes that traditionally defined common knowledge never arises in a BDTF. In contrast, we now use 
the Induction Rule and reachability arguments (\cref{induction-rule,reachability}) to show that under our definition,
common knowledge of all future signals in a BDTF always arises in finite time, analogously to the Getting to Common Knowledge Theorem of \citet[Section~7]{Geanakoplos94} (which analyzes a setting without timing frictions). This in particular implies that in the setting of \citet{GeanakoplosP82} with added timing frictions, common knowledge (as we have defined it) of posteriors arises in finite time under very mild conditions. 

We say that a BDTF has \emph{timestamps} if for each $i\in\allplayers$ and $(\history,t)\in\points$ such that $t\ge z^i_\history$, the signal $f_i(\history,t)$ identifies also the subjective time of $i$ at $(\history,t)$ as well as the subjective time---of the other player---that is identified in the signal just received by $i$ at $(\history,t)$ if any such signal was indeed received. One way to think about having timestamps is if both players communicate by email (so each email message contains the subjective time at which it is sent) and always send a new message by hitting ``reply'' on the last message that was received (so the subjective time at which that last message was sent is quoted).

For every player $i\in\allplayers$ and every $t\in T$, we use $[\tau_i=t]$ to denote the event ``the subjective time of $i$ is~$t$.'' We note that this is a singular, $i$-local event that occurs in every history.

\begin{theorem}[Getting to Common Knowledge with Timing Frictions]\label{get-ck}
Consider any BDTF with timestamps in which the set $O$ of initial conditions is finite.\footnote{The assumption that $O$ is finite is analogous to the assumption of \citet{Geanakoplos94} that the set of histories~$\histories$ is finite and to the assumption of \citet{GeanakoplosP82} that the number of kens in each player's initial knowledge partition over $\histories$ is finite (since it implies that $\histories$ is a union of a finite number of classes of forever-indistinguishable histories). Note, however, that in our setting both $\histories$ and the number of kens in each player's initial partition (and the number of classes of forever-indistinguishable histories) are infinite. We nonetheless prove that common knowledge arises in finite time.} For every history $\history$ there exist $\hat{t}_\alpha,\hat{t}_\beta\in\NN$ such that the contents of all signals to ever be sent by $\alpha$ after her subjective time $\hat{t}_\alpha$ and the contents of all signals to ever be sent by~$\beta$ after her subjective time $\hat{t}_\beta$ are common knowledge in $\history$ between $\alpha@[\tau_\alpha=\hat{t}_\alpha]$ and~$\beta@[\tau_\beta=\hat{t}_\beta]$.
\end{theorem}

The proof of \cref{get-ck} proceeds in several steps. The first step uses the Induction Rule (\cref{induction-rule}) to show that after a full round-trip of two sequential signals, there is common knowledge between the player who receives the latter of these two signals, as she receives it and sends her next signal, and the other player, as she receives this next signal, of the round-trip delay of signals as well as of upper and lower bounds on the players' birth-date difference. This is the step in which timestamps are used. The second step uses a reachability argument (\cref{reachability}) to show that the result of the first step implies that the relevant parts (in a precise sense) of the knowledge partitions at these events between which common knowledge holds all lie in a subset of $\histories$ that is finite modulo a certain equivalence relation on histories that is safe to gloss over. Finally, the third step proceeds with an argument reminiscent of that of \citet{Geanakoplos94} \citep[see also][]{GeanakoplosP82} to argue that the finiteness proven in the second step implies that after a finite number of steps all future signals become common knowledge. Crucially, this last step must be carefully carried out by partitioning $\histories$ not according to ``(absolute-)time slices'' as in previous papers, but rather using ``subjective-time slices,'' that is, partitioning~$\histories$ with respect to $i$'s \emph{subjective} time being kept constant (which corresponds to varying absolute times), and examining the refinement of such partitions as the respective subjective times of players progress.

\begin{proof}[Proof of \cref{get-ck}]
Let $\history$ be a history. Observe that in~$\history$, one of the players has a birth date that is no later than that of the other one. Without loss of generality, assume that this player is $\beta$, i.e., $z^\alpha_\history\ge z^\beta_\history$. (Note that we are \emph{not} assuming that either $\alpha$ or $\beta$ \emph{knows} that $\beta$'s birth date is not later than $\alpha$'s.) Therefore, any signal sent by~$\alpha$ is received by~$\beta$, since it arrives after $\beta$'s birth date. Let $t_1=z^\alpha_\history-z^\beta_\history+d^\beta_\history$ be the subjective time of player~$\beta$ at which in $\history$ she receives the signal that player~$\alpha$ sends at $\alpha$'s subjective time~$0$. Let $t_2=t_1+z^\beta_\history-z^\alpha_\history+d^\alpha_\history=d^\alpha_\history+d^\beta_\history$ be the subjective time of player~$\alpha$ at which in $\history$ she receives the signal that player~$\beta$ sends at $\beta$'s subjective time~$t_1$. Finally, let $t_3=t_2+z^\alpha_\history-z^\beta_\history+d^\beta_\history=t_1+t_2$ be the subjective time of player~$\beta$ at which in $\history$ she receives the signal that player~$\alpha$ sends at $\alpha$'s subjective time~$t_2$. 

Let $\psi_\alpha=[\tau_\alpha=t_2]$ (i.e., the $\alpha$-local event ``the subjective time of player~$\alpha$ is~$t_2$``), let $\psi_\beta=[\tau_\beta=t_3]$ (i.e., the $\beta$-local event ``the subjective time of player~$\beta$ is~$t_3$''), and denote $\bar{\psi}=(\psi_\alpha,\psi_\beta)$. 
Moreover, let~$\phi$ be the time-invariant event ``at $[\tau_\alpha=0]$ player~$\alpha$ sends a signal that is received by~$\beta$ at $[\tau_\beta=t_1]$, who immediately then sends a signal that is received by~$\alpha$ at $[\tau_\alpha=t_2]=\psi_\alpha$.'' 
We claim that due to timestamps, $\phi\subseteq E_{\allplayers\subat\bar{\psi}}\phi$. Indeed, to see that $\phi\subseteq K_{\alpha\subat\psi_\alpha}\phi$, note that if $\phi$ holds, then at $\psi_\alpha$ player~$\alpha$ knows that the signal that she sent at $[\tau_\alpha=0]$ was received by $\beta$ at $[\tau_\beta=t_1]$ due to the timestamps in the signal sent by $\beta$ at $[\tau_\beta=t_1]$ and received by $\alpha$ at $[\tau_\alpha=t_2]$ (i.e., when $\psi_\alpha$ holds). Similarly, to see that $\phi\subseteq K_{\beta\subat\psi_\beta}\phi$, note that if $\phi$ holds, then at $\psi_\beta$ player~$\beta$  knows that the signal that she sent at $[\tau_\beta=t_1]$ was received by $\alpha$ at $[\tau_\alpha=t_2]$ due to the timestamps in the signal sent by $\alpha$ at $[\tau_\alpha=t_2]$, which, since $\phi$ holds, must be received by $\beta$ at $[\tau_\beta=t_1+t_2]=[\tau_\beta=t_3]$ (i.e., when $\psi_\beta$ holds). Since $\phi\subseteq E_{\allplayers\subat\bar{\psi}}\phi$, by the Induction Rule (\cref{induction-rule}) we have that $\phi\subseteq C_{\allplayers\subat\bar{\psi}}\phi$ and so, in particular, $C_{\allplayers\subat\bar{\psi}}\phi$ holds in $\history$.

Recall that $t_2=d^\alpha_\history+d^\beta_\history$, i.e., $t_2$ is the round-trip delay in~$\history$. Let $\sigma_D$ be the time-invariant event ``the round-trip delay is $t_2$'' (i.e., ``$d^\alpha+d^\beta=t_2$''). Note that $\phi\subseteq\sigma_D$ (since the fact that $\phi$ holds in a history $\history'$ implies that $t_2=d^\alpha_{\history'}+d^\beta_{\history'}$), and hence by \cref{monotonicity} we have that $C_{\allplayers\subat\bar{\psi}}\phi\subseteq C_{\allplayers\subat\bar{\psi}}\sigma_D$, and so $C_{\allplayers\subat\bar{\psi}}\sigma_D$ holds in~$\history$.

Fix $Z\eqdef\max\{t_1,t_2-t_1\}$ and let $\sigma_Z$ be the time-invariant event ``the birth-date difference is smaller than $Z$'' (i.e., ``$|z^\alpha-z^\beta|<Z$''). Note that $\phi\subseteq\sigma_Z$ (since the fact that $\phi$ holds in a history $\history'$ implies that both $t_1=z^\alpha_{\history'}-z^\beta_{\history'}+d^\beta_{\history'}>z^\alpha_{\history'}-z^\beta_{\history'}$ and $t_2-t_1=z^\beta_{\history'}-z^\alpha_{\history'}+d^\alpha_{\history'}>z^\beta_{\history'}-z^\alpha_{\history'}$, which together imply that $|z^\alpha_{\history'}-z^\beta_{\history'}|<Z$), and hence by \cref{monotonicity} we have that $C_{\allplayers\subat\bar{\psi}}\phi\subseteq C_{\allplayers\subat\bar{\psi}}\sigma_Z$, and so $C_{\allplayers\subat\bar{\psi}}\sigma_Z$ holds in~$\history$.

For a singular event $\xi\subset\points$ and history $\history\in\histories$ in which $\xi$ occurs, we denote the unique time $t\in T$ such that $(\history,t)\in\xi$ by $t^\history_\xi$.
For every $i\in\allplayers$ and every singular, $i$-local event $\xi\subset\points$ that occurs in every history, let $P_{i\subat\xi}$ be the partition of $\histories$ such that $\history',\history''\in\histories$ are in the same partition cell if and only if $(\history',t^{\history'}_{\xi})\sim_i(\history'',t^{\history''}_{\xi})$, i.e., if and only if $i$ finds indistinguishable $\history'$ when $\xi$ occurs and $\history''$ when $\xi$ occurs. Of interest in the current proof are partitions of $\histories$ of the form $P_{i\subat[\tau_i=t]}$ for some $i\in\allplayers$ and~$t\in\NN$. In such a partition $P_{i\subat[\tau_i=t]}$, two histories~$\history',\history''\in\histories$ are in the same partition cell if and only if $i$ finds indistinguishable $\history'$ when $i$'s \emph{subjective time} is $t$ and $\history''$ when $i$'s \emph{subjective time} is~$t$.
Note that every cell of such a partition $P_{i\subat[\tau_i=t]}$ is a (disjoint) union of sets of histories of the form
\begin{equation}\label{equivalence}
\bigl\{\history''\in\histories\,~\big|~\,o_{\history''}=o_{\history'}~\,\&\,~d^\alpha_{\history''}=d^\alpha_{\history'}~\,\&\,~d^\beta_{\history''}=d^\beta_{\history'}~\,\&\,~z^\alpha_{\history''}-z^\beta_{\history''}=z^\alpha_{\history'}-z^\beta_{\history'}\bigr\}
\end{equation}
for some $\history'\in\histories$.

For any two partitions $P,P'$ of $\histories$, let $P\wedge P'$ denote the \emph{meet} of $P$ and $P'$, i.e., the finest common coarsening of $P$ and $P'$. Note that the partition cell of $\history$ in $P_{\alpha\subat\psi_\alpha}\wedge P_{\beta\subat\psi_\beta}$ is precisely the set of histories in the connected component of $\history$ in the reachability graph $G_{\allplayers,\bar{\psi}}$. Therefore, by \cref{reachability} and since both $C_{\allplayers\subat\bar{\psi}}\sigma_D$ and $C_{\allplayers\subat\bar{\psi}}\sigma_Z$ hold in $\history$, we have that $d^\alpha_{\history'}+d^\beta_{\history'}=t_2$ and $|z^\alpha_{\history'}-z^\beta_{\history'}|<Z$ for every history~$\history'$ in the partition cell of $\history$ in $P_{\alpha\subat\psi_\alpha}\wedge P_{\beta\subat\psi_\beta}$. Therefore, this partition cell is a \emph{finite} union of sets of histories of the form \eqref{equivalence}.

Consider the sequence of pairs of partitions $\bigl((P_{\alpha\subat[\tau_\alpha=t_2+\ell]},P_{\beta\subat[\tau_\beta=t_3+\ell]})\bigr)_{\ell=0}^\infty$. The partition pair corresponding to $\ell=0$ in this sequence is $(P_{\alpha\subat\psi_\alpha},P_{\beta\subat\psi_\beta})$ and, as $\ell$ grows, each of the two respective partitions in the pair becomes finer. Nonetheless, recall that for every $\history'\in\histories$, each set of histories of the form \eqref{equivalence} is always contained in its entirety in a single partition cell, in each of the  partitions $P_{\alpha\subat[\tau_\alpha=t_2+\ell]}$ and $P_{\beta\subat[\tau_\beta=t_3+\ell]}$. We therefore have that there exists $\ell'$ such that restricted to the partition cell of $\history$ in $P_{\alpha\subat\psi_\alpha}\wedge P_{\beta\subat\psi_\beta}$, the pair $(P_{\alpha\subat[\tau_\alpha=t_2+\ell]},P_{\beta\subat[\tau_\beta=t_3+\ell]})$ is constant for every $\ell\ge\ell'$. Setting $\hat{t}_\alpha\eqdef t_2+\ell'$ and $\hat{t}_\beta\eqdef t_3+\ell'$, it follows that all signals to ever be received (let alone sent) by each player $i\in\allplayers$ after her subjective time $\hat{t}_i$ are common knowledge in $\history$ between $\alpha@[\tau_\alpha=\hat{t}_\alpha]$ and $\beta@[\tau_\beta=\hat{t}_\beta]$.
\end{proof}

\section[Leveraging Common Knowledge: Agreement as a Case Study]{Leveraging Common Knowledge:\texorpdfstring{\\}{ }Agreement as a Case Study}\label{using}

In \cref{attaining}, we presented our definition of common knowledge and showed that it arises in settings in which we would expect common knowledge to hold even when traditionally defined common knowledge does not formally hold. In this section, we discuss the uses of (our notion) of common knowledge once it is ascertained to have arisen.

As noted in the introduction, it is quite straightforward to rederive in a dynamic context, using our notion of common knowledge, many results that are originally proved in a static context with traditionally defined common knowledge assumptions. For example, in a static context, \citet{BrandenburgerD87} proved that common knowledge of rationality and of the rules of the game implies rationalizability of the solution concept. Following an argument analogous to theirs, however focusing for each player $i$ on the partition~$P_{i\subat[\text{$i$ acts}]}$ (as defined in the proof of \cref{get-ck}), one can obtain that in a dynamic setting, common knowledge (as we define it) of rationality and of the rules of the game, between the different players as each of them acts, implies rationalizability of the solution concept. Note that while the assumption of common knowledge of rationality is the assumption that one usually attempts to weaken in this setting \citep{AumannB95}, it is the implied assumption of common knowledge of the rules of the game that we find to be especially strong.\footnote{\citet{FudenbergT91} write: ``Throughout our discussion, we assume that the structure of the game is common knowledge in an informal sense. Applying formal definitions of common knowledge to the structure of the game leads to technical and philosophical problems that we prefer not to address.''} Indeed, if some principal or mechanism designer sends a message with the rules of the game (or announces them in an online video meeting) to all players and there is even small temporal uncertainty in the delivery of the message, then by the Halpern--Moses Paradox the rules of the game will never become common knowledge, traditionally defined. In contrast, common knowledge of rationality might in fact hold \textit{a priori} (that is, at any point in time and does not need to dynamically arise) based on age-old conventions within a community.\footnote{\citet{ClarkM81} refer to such common knowledge as holding due to ``community membership.''} Our notion of common knowledge makes assuming common knowledge of the rules of the game far more realistic and natural.

As another example, \citet{MilgromS82} prove a no-trade theorem under the assumption of common knowledge of feasibility and individual rationality at the moment in which all traders agree to trade. Nowadays, most trades are performed in a computerized fashion and asynchronously confirmed, with individual traders possibly each clicking a button to agree to the trade at some different time. Similarly to the above discussion, the no-trade theorem of \citet{MilgromS82} can similarly be easily extended to this setting via an argument analogous to the original one, however focusing for each trader $i$ on the partition~$P_{i\subat[\text{$i$ agrees}]}$ (again, as defined in the proof of \cref{get-ck}), that is, assuming only common knowledge (as we define it) of individual rationality between the traders \emph{as each of them agrees to the trade}.

To demonstrate one such derivation of a canonical result with common knowledge assumptions for our notion of common knowledge, in this section we revisit the seminal ``Agreeing to Disagree'' theorem of \citet{Aumann76}. We focus on this result for several reasons: Because of its fundamental value; because it involves probabilistic analysis, whose adaptation to our model is somewhat less straightforward than the adaptation of some other results that we have mentioned; because it is known that any finite level of nested knowledge is insufficient for guaranteeing that this result holds \citep{GeanakoplosP82}; and, because the importance of analyzing this result in a dynamic setting was already established by \citet{GeanakoplosP82}.
An agreement theorem for our notion of common knowledge therefore also completes the picture with respect to the \citet{GeanakoplosP82} setting with added timing frictions, and establishes equality of posteriors in this setting despite the introduced timing frictions, which foil the attainability of traditionally defined common knowledge.

\subsection{An Agreement Theorem}
\label{sec:agreement}

While the setting considered by \citet{Aumann76} is static, in our case the players can enter the agreement at possibly different times, and hence we focus on agreement about time-invariant events. That is, we will be interested, e.g., in the posterior for a certain oil field containing a certain amount of oil rather than the posterior for whether it rains ``today.'' Indeed, if a player~$i$ sends a message on Monday to a player~$j$ and the message arrives on Tuesday, then even if the message contains all of $i$'s information, even if $j$ has no information before $i$ sends the message, and even if anything observed by either $i$ or $j$ while the message is in transit is also observed by the other, it might well be that $i$'s posterior, when she sends the message, for ``it rains today'' is~$1$ while $j$'s posterior, when she receives the message, for ``it rains today'' is~$0$ (and these are both common knowledge between $i$ as she sends the message and~$j$ as she receives it!), if it rains on Monday but the sun shines  on Tuesday.\footnote{Note that an event such as ``it rains in NYC at noon on Monday, February 8, 1971'' (the day the NASDAQ debuted) is time-invariant.} 

To properly reason about posteriors, we must equip our dynamic model with a probability space. We first emphasize that the probability measure in this space is not over (subsets of) $\points$ but rather over (subsets of) the set of histories~$\histories$. While this might seem surprising since~$\histories$ is not the same set over which our knowledge partitions are defined, this is analogous to \citet{GeanakoplosP82} defining the probability measure in their original setting over possible initial conditions. This choice of universe for our probability space is indeed suitable for reasoning about the probability of time-invariant events (as motivated in the previous paragraph), as these are \emph{history-dependent}: They either hold throughout an entire history or hold at no point of that history. 
In contrast, since knowledge evolves over time, our knowledge partitions must be defined over $\points$ rather than $\histories$. Moreover, as implicitly assumed by \citet{Aumann76}, we assume at most a countable number of kens in each player's knowledge partition.\footnote{\citet{GeanakoplosP82} require a finite number of kens in their analysis.}

For a concrete example of this framework, consider a set of histories that is induced by a finite number of possible initial facts about the world (i.e., possible ``initial conditions'' for a history), each with its own probability, as well as finitely many possible new (i.e., previously unobserved) facts about the world added at every time $t\in\Time$, each with its own probability conditioned on all facts about the world so far. Such new added facts might come from nature---e.g., whether a particular sent signal or message arrives in one or two steps---or from actions of players, e.g., when a player randomizes among several actions. This process, which in computer science terms might be called a \emph{branching process}, induces a probability measure over all possible histories. (It is unclear how one might use it to define a probability measure over~$\points$ in general without making further assumptions.) Even if all events are observed by all players as they occur in this process, the number of kens of each player's knowledge partition is at most countable---each ken contains all histories that agree on their prefix up to some time~$t$. If some events are not observed by all players as they happen, then the knowledge partitions become coarser and therefore could only contain fewer kens; in this case, too, knowledge partitions remain at most countable (and if time is finite, then each  knowledge partition is in fact finite).

To define the probability of a time-invariant event $\phi$, we denote  $\Pr(\phi)\eqdef\Pr\bigl(\histories(\phi)\bigr)$ (recall that $\histories(\phi)$ denotes the set of all histories in which $\phi$ holds), where the right-hand side is evaluated according to our probability measure over $\histories$. Similarly, if $\kappa\in\partition_i$ for some player~$i$ is a singular event (recall that singular events are those that occur at most once during any given history), then we define $i$'s \emph{posterior probability} for $\phi$ at every $(\history,t)\in\kappa$ as $\Pr(\phi\mid \kappa)\,\eqdef\,\Pr\bigl(\histories(\phi)~\big|~\histories(\kappa)\bigr)$, where the latter is again evaluated according to our probability measure over~$\histories$. For each of these definitions to be well-defined, the appropriate events in $\histories$ must be measurable according to our probability measure over~$\histories$ (when $\histories(\phi)$ is indeed measurable, by slight abuse of terminology we say that $\phi$ is measurable); for the posterior probability to be well-defined, we additionally must have $\Pr\bigl(\histories(\kappa)\bigr)>0$. Therefore, analogously to \citet{Aumann76} implicitly assuming that all kens are measurable and have positive probability (otherwise the posteriors in these kens are not well defined), we assume that the set $\histories(\kappa)$ is measurable and has positive probability for every singular ken (i.e., for every ken that is a singular event) $\kappa\in\partition_i$ for every $i\in\players$. (This is indeed the case, e.g.,  in the above branching-process example.) We will also use the fact that for every $q\in[0,1]$, the event ``$i$'s posterior for $\phi$ is $q$'' (i.e., the set of all $(\history,t)\in\points$ at which $i$'s posterior for $\phi$ is $q$) is an $i$-local event. (That is, for every $(\history,t)\in\points$, this event holds at $(\history,t)$ iff $i$ knows that it holds at $(\history,t)$.) We denote this event by $[\Pr_i(\phi)\!=\!q]$. 
For a singular $i$-local event $\psi_i$ we write $[\Pr_i(\phi)\!=\!q]@\psi_i\,\eqdef\,\sometime\bigl(\psi_i\cap[\Pr_i(\phi)\!=\!q]\bigr)$ to denote the time-invariant event ``$\psi_i$ holds in the current history, and when it does, $[\Pr_i(\phi)\!=\!q]$ holds as well.''
We are now ready to state and prove our agreement theorem.

\begin{theorem}[Agreement]\label{agreement}
Let $\players=\{\alpha,\beta$\} be a set of two players, let $\phi$ be a measurable time-invariant event, and let $\bar{\psi}=(\psi_\alpha,\psi_\beta)$ be an $\players$-profile of singular events. If, for some $q_\alpha,q_\beta\in[0,1]$, it is the case that
$C_{\players\subat\bar{\psi}}\bigl([\Pr_\alpha(\phi)\!=\!q_\alpha]@\psi_\alpha\cap[\Pr_\beta(\phi)\!=\!q_\beta]@\psi_\beta\bigr)
\ne\emptyset$, then $q_\alpha=q_\beta$.
\end{theorem}

In \cref{agreement}, the condition that the common knowledge event is nonempty means that there exists a history at which it is common knowledge (under our definition) between $\alpha@\psi_\alpha$ and $\beta@\psi_\beta$ that the former's posterior is $q_\alpha$ and the latter's posterior is $q_\beta$. \cref{agreement} guarantees that under this condition, necessarily $q_\alpha=q_\beta$.

\begin{proof}
Denote $C\eqdef C_{\players\subat\bar{\psi}}\bigl([\Pr_\alpha(\phi)\!=\!q_\alpha]@\psi_\alpha\cap[\Pr_\beta(\phi)\!=\!q_\beta]@\psi_\beta\bigr)$, let $i\in\players$, and denote $C^i\eqdef\psi_i\cap C$.
By \cref{cki}(1), we have that $\histories(C^i)=\histories(C)\ne\emptyset$. By \cref{ck-local}, the event $C^i$ is $i$-local, and hence it is of the form $C^i=\bigcup_{\ell\in L}\kappa^i_\ell$, where $\{\kappa^i_\ell\}_{\ell\in L}$ is a (nonempty) set of kens of the partition~$\partition_i$. Since $\psi_i$ is singular, so is~$C^i$, and hence the kens in~$\{\kappa^i_\ell\}_{\ell\in L}$ are singular and have pairwise-disjoint history sets. By \cref{cki}(2) and by definition of a posterior, $q_i\cdot\Pr\bigl(\histories(\kappa^i_\ell)\bigr)=\Pr\bigl(\histories(\phi)\cap \histories(\kappa^i_\ell)\bigr)$ holds (with $\Pr\bigl(\histories(\kappa^i_\ell)\bigr)>0$) for every~$\ell\in L$. Summing over $\ell\in L$, and noting that $L$ is at most countable by the assumption of at most countably many kens in~$\partition_i$, we therefore have that $q_i\cdot\Pr\bigl(\histories(C^i)\bigr)=\Pr\bigl(\histories(\phi)\cap \histories(C^i)\bigr)$ with $\Pr\bigl(\histories(C^i)\bigr)>0$, and recall that this holds for every $i\in\{\alpha,\beta\}$. Since $\histories(C^\alpha)=\histories(C)=\histories(C^\beta)$, it follows that $q_\alpha=q_\beta$, as claimed.
\end{proof}

By \cref{agreement}, if the posteriors of $\alpha$ at $\psi_\alpha$ and $\beta$ at $\psi_\beta$ are common knowledge between the former and the latter, then these posteriors must coincide. The proof follows the same general structure (appropriately modified) as the original proof in \citet{Aumann76}, which demonstrates that our definition of common knowledge naturally yields similar consequences to those of the traditional one. Since, as discussed, a consequence of \cref{get-ck} is that in the \citet{GeanakoplosP82} setting with added timing frictions, there exist appropriate events $\psi_\alpha$ and $\psi_\beta$ such that common knowledge of the respective posteriors arises between $\alpha$ and $\beta$, by \cref{agreement} these posteriors must coincide, establishing the central result of \citet{GeanakoplosP82} even in the presence of timing frictions, even though common knowledge of posteriors (as traditionally defined) is never attained.

\section[Characterizing Equilibria via Common Knowledge]{Characterizing Equilibria via\texorpdfstring{\\}{ }Common Knowledge}\label{game}

In this \lcnamecref{game}, we demonstrate the usefulness of our notion of common knowledge for characterizing equilibrium behavior in a family of dynamic Bayesian coordination games that we call \emph{coordinated-attack games}. In a coordinated-attack game, each player may decide whether and when to initiate an attack, and players are better off initiating attacks only if the state of nature is such that an attack would be successful, and only if each of the players initiates her attack sufficiently early. If one of the players initiates an attack but the other does not do so early enough (or at all), then both players suffer a loss (due to their joint army's forces diminishing or being wiped out completely). Initiating an attack early carries risk (both regarding the other player's action and regarding the state of nature), but also has the potential for reward. Our characterization of equilibrium behavior in these games holds regardless of the specific technology by which the players learn about the state of nature and by which they communicate, 
despite different such technologies resulting in seemingly very disparate equilibria. This is achieved by stating the characterization in terms of the players' state of knowledge using our notion of common knowledge. We start by formally defining this family of games.

\subsection{Coordinated-Attack Games}

\paragraph{Players and states of nature} A \emph{coordinated-attack game} is a dynamic Bayesian game played over 100 periods, denoted $t=0,\ldots,99$. The set of players is $\allplayers=\{\alpha,\beta\}$. A \emph{state of nature}~$o=\bigl(q_o,(z^i_o)_{i\in\allplayers},(t^i_o)_{i\in\allplayers},(d^{i,t}_o)_{i\in\allplayers}^{t=0,\ldots,99}\bigr)$ in a coordinated-attack game consists of an \emph{attack prospect} $q_o\in\{0,1\}$, a \emph{birth date} $z^i_o\in\{0,\ldots,100\}$ and an \emph{observation time} $t^i_o\in\{0,\ldots,100\}$ for each player $i\in\allplayers$, and a \emph{message delay} $d^{i,t}_o\in\{1,\ldots,100\}$ for each player $i\in\allplayers$ and time $t=0,\ldots,99$. The set of possible states of nature is denoted~$O$, and a specific coordinated-attack game is defined by a common Bayesian prior over~$O$. We emphasize that this Bayesian prior need not be a product distribution, and might not have full support.

\paragraph{Game play} Given a state of nature~$o\in O$, play in a coordinated-attack game proceeds as follows. At each time $t=0,\ldots,99$, each player $i\in\allplayers$ whose birth date satisfies $t\ge z^i_o$ is called to action. Such a player $i$ chooses an action, which can depend on (1)~$i$'s \emph{subjective time} $t-z^i_o$; (2)~the content and subjective time of receipt $t'+d^{i,t'}_o-z^i_o$ of each message sent by the other player at any time $t'$ such that $t\ge t'+d^{i,t'}_o\ge z^i_o$; and (3)~if $t\ge t^i_o\ge z^i_o$, both the subjective observation time $t^i_o-z^i_o$ and the attack prospect $q_o$.\footnote{Observe that a birth date of $100$ indicates that the player is never called to action, a delay of $100$ indicates the relevant message is never received, and an observation time of $100$ indicates that the player never observes the attack prospect.} An action consists of both a message (a finite sequence of characters) to send to the other player at the current time $t$ and an \emph{attack decision}. At the onset, possible attack decisions are in $\{\text{``initiate attack''},\text{``do not  attack''}\}$, and once player $i$ has initiated an attack, in following time periods she has only one possible attack decision, called $\text{``already attacking''}$. That is, each player may initiate an attack at most once, and initiating an attack is binding (e.g., because it reveals the location of one's battalion).

\paragraph{Payoffs} Recall that the game runs until (absolute) time $t=99$. We say that \emph{the attack is successful} in a given run of the game in which some player attacks if (1)~player $\alpha$ initiates an attack by (absolute) time $\hat{t}_\alpha\eqdef49$, (2)~player $\beta$ initiates an attack by (absolute) time $\hat{t}_\beta\eqdef99$, and (3)~the attack prospect (recall that this is a part of the state of nature) is $1$.\footnote{We fix the values of $\hat{t}_\alpha$ and $\hat{t}_\beta$ (and the number of rounds in the game) for concreteness. Nothing in our analysis depends on the specific choice of these values.} If neither player initiates an attack, then each player receives utility~$0$. If some player initiates an attack, then each player receives utility~$1$ if the attack is successful and utility $U<0$ otherwise.

\bigskip

Before we commence our analysis of equilibrium behavior in coordinated-attack games, we demonstrate the diverse structure of equilibria in different coordinated-attack games via several examples. We start with a coordinated-attack game that has an equilibrium of a particularly simple form.

\begin{example}\label{example-one}
Consider a coordinated-attack game in which (1) each attack prospect $q\in\{0,1\}$ is equally likely, (2) player $\alpha$ is born at objective time $0$ and immediately observes the attack prospect~$q$, (3) player $\beta$ is born at an independently and uniformly drawn objective time $z\in\{0,1\}$ and never (directly) observes the attack prospect, and (4)~messages sent by player $\alpha$ take $50\!+\!z$ periods to arrive and messages sent by player~$\beta$ take $51\!-\!z$ periods to arrive. In this coordinated-attack game, the welfare-maximizing Nash equilibrium results in a successful attack in \emph{every} history in which the attack prospect is $1$, and in no attacks when the prospect is~$0$. This equilibrium (up to trivial degrees of freedom) is as follows. Player $\alpha$ sends the attack prospect to $\beta$ upon observing it (or more generally, by time $48$, which is necessary and sufficient to guarantee that this message reaches $\beta$ by $\beta$'s deadline for initiating a successful attack, $\hat{t}_\beta=99$) and furthermore, initiates an attack (no later than at her deadline for initiating a successful attack, $\hat{t}_\alpha=49$) if and only if 
she observes that
the attack prospect is~$1$. Player $\beta$ initiates an attack (no later than at her deadline, $\hat{t}_\beta=99$) if and only if 
she receives a message from $\alpha$ stating that the attack prospect is $1$.
\end{example}

We continue with two coordinated-attack games in which equilibria have more subtle structures.

\begin{example}\label{example-consecutive}
Consider a coordinated-attack game in which (1) each attack prospect is equally likely, (2) player $\alpha$ is born at objective time $0$ and immediately observes the attack prospect, (3) player $\beta$ is born at an independently and uniformly drawn objective time $z\in\{0,1\}$ and never (directly) observes the attack prospect, (4) for a value $d\in\{1,\ldots,100\}$ drawn once uniformly and independently, each message sent by player $\alpha$ takes $d$ periods to arrive, and (5) messages sent by player~$\beta$ take $3\!-\!z$ periods to arrive. In this case, the welfare-maximizing equilibrium is as follows. Player $\alpha$ sends the attack prospect to $\beta$ immediately upon observing it. Player $\beta$ notifies $\alpha$ once she receives a message with the attack prospect and furthermore, initiates an attack (by $\hat{t}_\beta$) if and only if she receives such a message indicating that the attack prospect is $1$ at a subjective time not greater than~$46$ (this is necessary and sufficient for $\beta$'s return message to reach $\alpha$ by $\alpha$'s deadline $\hat{t}_\alpha=49$). Player~$\alpha$, in turn, initiates an attack (by $\hat{t}_\alpha$) if and only if she receives a message from $\beta$ stating that she received $\alpha$'s message stating that the attack prospect is $1$ no later than $\beta$'s subjective time~46. In this coordinated-attack game, two consecutive messages (of a particular form) are required for an initiated attack to be guaranteed to succeed.
\end{example}

\begin{example}\label{example-parallel}
Consider a coordinated-attack game in which (1) each attack prospect is equally likely, (2) player $\alpha$ is born at objective time $0$ and immediately observes the attack prospect, (3) player $\beta$ is born at an independently and uniformly drawn objective time $z\in\{0,100\}$ and never (directly) observes the attack prospect, and (4) messages sent by either player at its subjective time $0$ take one period to arrive, and messages sent at later times take $100$ periods to arrive. In this case, the welfare-maximizing equilibrium is as follows. Player $\alpha$ sends the attack prospect to $\beta$ immediately upon observing it, and player $\beta$ sends the message ``I'm alive!'' to $\alpha$ immediately upon being born. Player~$\alpha$ initiates an attack (by $\hat{t}_\alpha$) if and only if the attack prospect is $1$ and she receives an ``I'm alive!'' message from~$\beta$, and $\beta$ initiates an attack (by $\hat{t}_\beta$) if and only if she receives a message from $\alpha$ indicating that the attack prospect is $1$. In this coordinated-attack game, two \emph{non}-consecutive messages are required for an initiated attack to be guaranteed to succeed.
\end{example}

 The welfare-maximizing equilibria in these three examples (and especially in the latter two) might seem at first glance to be driven by qualitatively different effects. It is therefore unclear how one might generally characterize welfare-maximizing equilibria in a manner that captures all three examples, let alone captures equilibria in all coordinated-attack games. Nonetheless, the main result of this section is a unified characterization of welfare-maximizing equilibria in \emph{all} coordinated-attack games. Notably, this is achieved by stating the characterization using our new notion of common knowledge. 

\subsection{Equilibrium Characterization}

When analyzing a coordinated-attack game, it will be convenient to first define a \emph{message strategy} for each player, which determines the content of sent messages, and then augment the message strategies with an \emph{attack strategy} for each player, which specifies attack decisions based on all information observed by the player so far. We start by defining a pair of message strategies  to which we refer as the \emph{send-all} strategies. For each player $i$, the message strategy $s_i^{\text{send-all}}$ sends, at every subjective time $t$, the subjective time, the content and subjective time of receipt of each message received by $i$ so far, and if $i$ has observed the attack prospect, then the attack prospect and its subjective observation time.

We now formalize the knowledge partition of each player when both players use their send-all strategies (regardless of their attack strategies). Fix a coordinated-attack game, i.e., a Bayesian prior $F\in\Delta(O)$ for both players.
Let $\histories\eqdef\supp(F)$. For every point $\bigl(\history,t\bigr)\in\points$ and player $i\in\allplayers$, let $M^i\bigl(\history,t\bigr)$ be the set consisting, for every message sent by the other player at any time $t'$ such that $t\ge t'+d^{i,t'}_\history\ge z^i_\history$ (when the players play according to the strategy profile $(s_\alpha^{\text{send-all}},s_\alpha^{\text{send-all}})$), of the pair of the message and its subjective time of receipt $t'+d^{i,t'}_\history-z^i_\history$.
For every player $i\in\allplayers$ and every pair of points $(\history,t),(\history',t')\in\points$, the partition $\partition_i$ of player $i$ over $\points$ satisfies that $(\history,t)\sim_i(\history',t')$ if and only if either of the following holds:
\begin{itemize}
\item 
$t-z^i_\history<0$ and $t'-z^i_{\history'}<0$. (Player~$i$ cannot distinguish between points prior to her birth date.)
\item
$t-z^i_\history=t'-z^i_{\history'}\ge0$ and all of the following hold:
\begin{itemize}
    \item $M^i(\history,t)=M^i(\history',t'\bigr)$ (same messages received by $i$) and
    \item either of the following two hold:
    \begin{itemize}
    \item $t^i_\history-z^i_\history=t^i_{\history'}-z^i_{\history'}$ and $t^i_\history\le t$ (and $t^i_{\history'}\le t'$) and $q_\history=q_{\history'}$. \\ (Same attack prospect observation by~$i$.)
    \item $t^i_\history>t$ and $t^i_{\history'}>t'$. (Attack prospect not yet observed by $i$.)
    \end{itemize}
\end{itemize}
\end{itemize}
We note that attack strategies for player $i$ that augment her send-all message strategy are well defined if and only if they are measurable with respect to this knowledge partition.

We henceforth focus on the case of $U=-\infty$, i.e., if either player initiates attacks but the attack is not successful, then this player's battalion is destroyed, and absent this battalion, the army to which both players belong loses the war. We show that in this setting, our notion of common knowledge characterizes equilibrium behavior in any coordinated-attack game, regardless of the specific technology (i.e., distribution over message delays and observation times) by which the players learn about the attack prospect and communicate in the game.
For an event $\psi$ and time $\hat{t}\in\{0,\ldots,99\}$, we denote by $[t_\psi\!\le\!\hat{t}\,]$ the time-invariant event ``At some time in the current history before or at time $\hat{t}$, the event $\psi$ holds.'' For an attack prospect $\hat{q}\in\{0,1\}$, we denote the time-invariant event ``the attack prospect is $\hat{q}$\,'' by $[q\!=\!\hat{q}]$. For each player $i$, we define the strategy $s^{\text{CK}}_i$ to be the strategy in which $i$ sends messages as dictated by $s_i^{\text{send-all}}$ and initiates an attack in the first period at which (with respect to $\points,\partition_\alpha,\partition_\beta$ as defined above) for some $\allplayers$-profile $\bar{\psi}=(\psi_\alpha,\psi_\beta)$, the event
\[
C^i_{\allplayers\subat\bar{\psi}}\bigl([q\!=\!1]\cap[t_{\psi_\alpha}\!\le\!\hat{t}_\alpha]\cap[t_{\psi_\beta}\!\le\!\hat{t}_\beta]\bigr)
\]
holds. (Player $i$ does not initiate an attack according to this strategy in histories in which this event never holds.) That is, each player $i\in\{\alpha,\beta\}$ initiates an attack as soon as for some such~$\bar{\psi}=(\psi_\alpha,\psi_\beta)$ it is the case that $i$'s ``individualized part'' of common knowledge between $\alpha@\psi_\alpha$ and $\beta@\psi_\beta$ of the following two facts holds: (1) the attack prospect is $1$, and (2) each player's part of~$\bar{\psi}$ occurs early enough for this player to still be able to initiate an attack in time at that point. By \cref{ck-local}, this attack strategy is measurable with respect to $i$'s knowledge partition. (We define off-path behavior, i.e., behavior if $i$ receives messages inconsistent with $s_{-i}^{\text{send-all}}$ being played, arbitrarily.)

Our main result shows that the attainment of common knowledge (under our definition) characterizes equilibrium behavior in this game.  We say that a strategy profile \emph{never results in an unsuccessful attack} if, when this strategy profile is played, the probability that a player initiates an attack and yet the attack is not successful is zero. We say that one strategy profile \emph{Pareto dominates} another if for each state of nature, the former achieves weakly higher (expected) welfare than the latter.

\begin{theorem}\label{pd-equilibrium}
In every coordinated-attack game with $U=-\infty$, the strategy profile $(s^{\text{CK}}_\alpha,s^{\text{CK}}_\beta)$ is a welfare-maximizing Nash equilibrium. This Nash equilibrium never results in an unsuccessful attack, and Pareto dominates every other strategy profile that never results in an unsuccessful attack.    
\end{theorem}

\cref{pd-equilibrium} implies that strictly positive welfare at equilibrium is possible if and only if common knowledge (as we define it) of the pertinent facts is attainable.

\begin{corollary}\label{positive-utility-iff-ck}
In every coordinated-attack game with $U=-\infty$, there exists a Nash equilibrium with strictly positive welfare if and only if there exists an $\allplayers$-profile $\bar{\psi}=(\psi_\alpha,\psi_\beta)$ such that (with respect to $\points,\partition_\alpha,\partition_\beta$ as defined above) it is the case that $C_{\allplayers\subat\bar{\psi}}\bigl([q\!=\!1]\cap[t_{\psi_\alpha}\!\le\!\hat{t}_\alpha]\cap[t_{\psi_\beta}\!\le\!\hat{t}_\beta]\bigr)
\ne\emptyset$.
\end{corollary}

The strength and generality of \cref{pd-equilibrium} is best appreciated by recalling \cref{example-parallel,example-consecutive}, which as we noted, showcase that it is unclear how one might directly (mechanically) characterize both of the equilibria in these examples rather than through epistemic notions. By showing that strategies that are based on our notion of common knowledge yield a unified characterization of these two equilibria, \cref{pd-equilibrium} highlights the benefits of our epistemic approach. Of course, there are coordinated-attack games in which the welfare-maximizing equilibrium is even more involved than in these examples. And yet, by leveraging our new notion of common knowledge, \cref{pd-equilibrium} characterizes welfare-maximizing equilibria in \emph{all} of these games. The proof of
\cref{pd-equilibrium} is based on the following three \lcnamecrefs{simulation}.

\begin{lemma}\label{never-unsuccessful}
Play according to $(s^{\text{CK}}_\alpha,s^{\text{CK}}_\beta)$ never results in an unsuccessful attack.
\end{lemma}

\begin{proof}
For every $N$-profile $\bar{\psi}$ and $i\in\allplayers$,
denote $C_{\bar{\psi}}\eqdef C_{\allplayers\subat\bar{\psi}}\bigl([q\!=\!1]\cap[t_{\psi_\alpha}\!\le\!\hat{t}_\alpha]\cap[t_{\psi_\beta}\!\le\!\hat{t}_\beta]\bigr)$ and $C_{\bar{\psi}}^i\eqdef\psi_i\cap C_{\bar{\psi}}$. Recall that for each $i$, the strategy $s^{\text{CK}}_i$ initiates an attack if and only if it is the first time in the history at which $C_{\bar{\psi}}^i$ holds for some $\bar{\psi}$. By \cref{cki}(1), $\histories(C_{\bar{\psi}}^i)=\histories(C_{\bar{\psi}})$ for every $i$ and $\bar{\psi}$. Therefore, both players initiate an attack in each history in $\cup_{\bar{\psi}}\histories(C_{\bar{\psi}})$, and neither player initiates an attack in any other history. Furthermore, by \cref{cki}(2), for every $\bar{\psi}$, in each of the histories~$\histories(C_{\bar{\psi}})$ the attack prospect is~$1$, and $\psi_i$---and hence $C_{\bar{\psi}}^i$---occurs before or at time $\hat{t}_i$ for every $i\in\allplayers$. Thus, whenever one initiates an attack, the attack is successful as required.
\end{proof}

\begin{lemma}\label{success-iff-ck}
Let $(s_\alpha,s_\beta)$ be a profile of pure strategies in which each player~$i$ sends messages as dictated by $s_i^{\text{send-all}}$. Play according to $(s_\alpha,s_\beta)$ never results in an unsuccessful attack if and only if there exists an $\allplayers$-profile $\bar{\psi}=(\psi_\alpha,\psi_\beta)$ such that each player $i\in\allplayers$ initiates an attack if and only if the event
$C^i_{\allplayers\subat\bar{\psi}}\bigl([q\!=\!1]\cap[t_{\psi_\alpha}\!\le\!\hat{t}_\alpha]\cap[t_{\psi_\beta}\!\le\!\hat{t}_\beta]\bigr)$
holds.
\end{lemma}

\begin{proof}
The proof of the ``if'' direction is very similar to the proof of \cref{never-unsuccessful} and is not used in our analysis, so we leave it to the reader.

For the ``only if'' direction, assume that  $(s_\alpha,s_\beta)$ never results in an unsuccessful attack. Let $\psi_i$ be the event ``$i$ is initiating an attack (right now).'' By measurability of the attack strategy, this is an $i$-local event, and hence $\bar{\psi}=\{\psi_\alpha,\psi_\beta\}$ is an $\allplayers$-profile. Let $\phi\eqdef[q\!=\!1]\cap[t_{\psi_\alpha}\!\le\!\hat{t}_\alpha]\cap[t_{\psi_\beta}\!\le\!\hat{t}_\beta]$. Since $(s_\alpha,s_\beta)$ never results in an unsuccessful attack, we have that (1)~$\bar{\psi}$ satisfies co-occurrence in~$\histories$ and (2)~$\psi_i\subseteq\phi$ for every $i\in\allplayers$, and since $\phi$ is time-invariant, also $\sometime \psi_i\subseteq\phi$. By \cref{cooccurrence-iff}, it follows for every $i\in\allplayers$ that $\psi_i\subseteq C_{\allplayers\subat\bar{\psi}}\phi$ and therefore $\psi_i=\psi_i\cap C_{\allplayers\subat\bar{\psi}}\phi=C^i_{\allplayers\subat\bar{\psi}}\phi$, as required.
\end{proof}

\begin{lemma}\label{simulation}
For every profile $(s_\alpha,s_\beta)$ of pure strategies, there exists a profile $(s'_\alpha,s'_\beta)$ of pure strategies such that the following holds for every player $i\in\allplayers$:
\begin{itemize}
    \item For each state of nature $o$ and time $t$, player $i$ initiates an attack at $t$ in $o$ when $(s'_\alpha,s'_\beta)$ is played if and only if $i$ initiates an attack at $t$ in $o$ in when $(s_\alpha,s_\beta)$ is played.
    \item In the strategy $s'_i$, player~$i$ sends messages as dictated by $s_i^{\text{send-all}}$.
\end{itemize}
\end{lemma}

\begin{proof}
We start by proving that for every state of nature $o$ and time $t$, the message that player $i$ sends at time $t$ in $o$ when messages are sent as dictated by $(s_\alpha^{\text{send-all}},s_\beta^{\text{send-all}})$ uniquely determines the message that $i$ sends at time $t$ in $o$ when $(s_\alpha,s_\beta)$ is played. We prove this for every fixed $o$, by full induction over $t$. Assume that the claim holds for some $o$ and all $t'<t$. Therefore, all messages received by $i$ up until and including time $t$ (which are all sent at times strictly prior to $t'$) when messages are sent as dictated by $(s_\alpha^{\text{send-all}},s_\beta^{\text{send-all}})$ uniquely determine the corresponding messages when $(s_\alpha,s_\beta)$ is played. Hence, the message sent by $i$ at $t$ in the former setting (which in particular contains all of these prior messages) completely determines everything observed by $i$ up until time $t$ in the latter setting, and hence, since $s_i$ is a pure strategy, determines $i$'s message according to $s_i$, concluding the inductive argument.

We define the strategy $s'_i$ by augmenting the message strategy $s_i^{\text{send-all}}$ with the following attack strategy. For every $o$  and $t$, we define the on-path attack action of~$s'_i$ (when messages are sent according to $(s_\alpha^{\text{send-all}},s_\beta^{\text{send-all}})$) at $t$ in $o$ to equal the attack action of $s_i$ at $t$ in $o$ when $(s_\alpha,s_\beta)$ is played. This defines a measurable attack strategy since by the claim just proven (applied to the other player), $i$'s knowledge when messages are sent according to $(s_\alpha^{\text{send-all}},s_\beta^{\text{send-all}})$ uniquely determines $i$'s knowledge when $(s_\alpha,s_\beta)$ is played, since every message received by $i$ in the latter setting is uniquely determined by the corresponding message received by $i$ in the former setting. We define off-path attack behavior in $s'_i$ arbitrarily.
\end{proof}

\begin{proof}[Proof of \cref{pd-equilibrium}]
By \cref{never-unsuccessful}, $(s^{\text{CK}}_\alpha,s^{\text{CK}}_\beta)$ never results in an unsuccessful attack. We start by proving that $(s^{\text{CK}}_\alpha,s^{\text{CK}}_\beta)$ Pareto dominates every other strategy profile that never results in an unsuccessful attack.

Let $(\sigma_\alpha,\sigma_\beta)$ be a (possibly mixed) strategy profile that never results in an unsuccessful attack. We will show that $(s^{\text{CK}}_\alpha,s^{\text{CK}}_\beta)$ Pareto dominates $(\sigma_\alpha,\sigma_\beta)$ by showing that the former Pareto dominates each pure strategy profile in the support of the latter. Let $(s_\alpha,s_\beta)$ be such a pure strategy profile; note that $(s_\alpha,s_\beta)$ never results in an unsuccessful attack.
By \cref{simulation}, there exists a pure strategy profile $(s'_\alpha,s'_\beta)$ in which each player~$i$ sends messages as dictated by $s_i^{\text{send-all}}$ and furthermore, each player initiates an attack in $(s'_\alpha,s'_\beta)$ whenever she initiates an attack in $(s_\alpha,s_\beta)$. By \cref{success-iff-ck}, we have that there exists an $\allplayers$-profile $\bar{\psi}=(\psi_\alpha,\psi_\beta)$ such that each player $i$ initiates an attack in $(s'_\alpha,s'_\beta)$ in the first period at which the event
$C^i_{\allplayers\subat\bar{\psi}}\bigl([q\!=\!1]\cap[t_{\psi_\alpha}\!\le\!\hat{t}_\alpha]\cap[t_{\psi_\beta}\!\le\!\hat{t}_\beta]\bigr)$
holds (and does not initiate an attack if this event never holds). Therefore, whenever an attack is initiated in $(s'_\alpha,s'_\beta)$, it is also initiates in $(s^{\text{CK}}_\alpha,s^{\text{CK}}_\beta)$, and since by \cref{never-unsuccessful} the latter never results in an unsuccessful attack, we have that the latter Pareto dominates the former as required.

We now prove that $(s^{\text{CK}}_\alpha,s^{\text{CK}}_\beta)$ is a Nash equilibrium. Assume by way of contradiction that there exists a profitable deviation $s_i$ for some player $i$. By \cref{never-unsuccessful}, the expected utility for each player when $(s^{\text{CK}}_\alpha,s^{\text{CK}}_\beta)$ is played is nonnegative. Therefore, the expected utility of player $i$ from playing $(s_i,s^{\text{CK}}_{-i})$ is nonnegative; hence, $(s_i,s^{\text{CK}}_{-i})$ never results in an unsuccessful attack. Therefore, $(s^{\text{CK}}_\alpha,s^{\text{CK}}_\beta)$ Pareto dominates $(s_i,s^{\text{CK}}_{-i})$, and hence $i$'s expected utility in $(s^{\text{CK}}_\alpha,s^{\text{CK}}_\beta)$ is weakly greater than in $(s_i,s^{\text{CK}}_{-i})$---a contradiction to $s_i$ being a profitable deviation for player $i$.

Since $(s_\alpha,s_\beta)$ is a Nash equilibrium that never results in an unsuccessful attack and Pareto dominates every other strategy profile that never results in an unsuccessful attack, $(s_\alpha,s_\beta)$ is a welfare-maximizing Nash equilibrium, as claimed.
\end{proof}

\begin{proof}[Proof of \cref{positive-utility-iff-ck}]
By \cref{pd-equilibrium}, there exists a Nash equilibrium with strictly positive expected welfare if and only if $(s^{\text{CK}}_\alpha,s^{\text{CK}}_\beta)$ achieves strictly positive expected welfare. Since $(s^{\text{CK}}_\alpha,s^{\text{CK}}_\beta)$ never results in an unsuccessful attack, it yields strictly positive expected welfare if and only if an attack is initiated in at least one history when this Nash equilibrium is played. By definition of $(s^{\text{CK}}_\alpha,s^{\text{CK}}_\beta)$, this occurs if and only if $C_{\allplayers\subat\bar{\psi}}\bigl([q\!=\!1]\cap[t_{\psi_\alpha}\!\le\!\hat{t}_\alpha]\cap[t_{\psi_\beta}\!\le\!\hat{t}_\beta]\bigr)
\ne\emptyset$.
\end{proof}

\cref{pd-equilibrium} shows that in a coordinated-attack game, the welfare-maximizing equilibrium features a cooperative effort to sufficiently inform one another so as to give rise to common knowledge as we define it, which is necessary and sufficient for initiating a successful attack. Crucially, this is true regardless of the specific technology by which the players learn about the attack prospect and communicate in any specific coordinated-attack game of interest.

We emphasize that equilibrium behavior in coordinated-attack games cannot in general be characterized by common knowledge as traditionally defined. Indeed, recalling the coordinated-attack game from \cref{example-one}, we note that by \cref{no-ck} (see also the discussion that follows the proof of that theorem), common knowledge as traditionally defined of the attack prospect is never attained in that coordinated-attack game. Nonetheless, the welfare-maximizing Nash equilibrium results in a successful attack in \emph{every} history in which the attack prospect is $1$. Moreover, whenever $\alpha$ initiates an attack in this equilibrium, $\beta$ does not yet even know that the attack prospect is $1$; indeed, $\beta$ only learns the attack prospect after time $\hat{t}_\alpha$, i.e., at a time at which it is too late for $\alpha$ to initiate a successful attack.

Finally, we remark that for the case of $U>-\infty$, an analog of \cref{pd-equilibrium} (and \cref{positive-utility-iff-ck}) holds, which instead of being based upon our modified notion of common knowledge, is based upon a notion of common $p$ belief \citep{MondererS89}, similarly modified.

\section{Discussion}\label{discussion}

In this paper, we resolve what is possibly the oldest open question at the interface of economics and computer science. What started as a push to apply epistemic analysis to dynamic systems starting in the seminal paper of \citet{HalpernM90} within the computer science literature, can be seen as coming full circle in this paper back into economic theory, and allowing to distill a better understanding of the essence of common knowledge and to show that simultaneity is not an inseparable part thereof but rather somewhat of a red herring.

One possible approach to sidestepping the lack of (traditionally defined) common knowledge in many settings might be to seek refuge in the theoretical model being a stylized version of reality. Indeed, even if a fact is not common knowledge, one might be tempted to perform the theoretical analysis as if this fact were common knowledge, and claim that the result should still be essentially correct. But it is unclear how far one might be able to stretch this ``as if,'' especially since papers such as \citet{SteinerS11} and \citet{Morris14} abound with counterintuitive results that pop up if one does not go along with such an ``as if,'' and show that common knowledge, as traditionally defined, is quite fragile. Our notion of common knowledge facilitates leveraging well-known consequences of common knowledge without the need for an imprecise ``as if,'' even in economic settings in which the main results of the above papers show that common knowledge, under its traditional definition, never arises. Our definition and analysis thus unearth that common knowledge (and with is, its celebrated implications) is considerably less fragile than previously believed.

While our definition of common knowledge overcomes timing frictions, it does not aid with attaining common knowledge in settings with probabilistically successful deliveries (as in \citealp{Rubinstein89}, and as in the case of the ``defective'' deliveries of \citealp{SteinerS11}). We do not view this as a shortcoming: Indeed, since unreliable deliveries are known not only to lose us the ability to harness common knowledge, but also to cause \emph{consequences} of common knowledge to formally fail, any epistemic notion that nevertheless holds in such situations certainly should not be called common knowledge. Probabilistic notions such as common $p$ belief remain an important tool for analyzing these and other settings. As we have demonstrated, extending such concepts using similar ideas to those that we develop in this paper for common knowledge seems promising as well, and future work might look to explore such extensions in greater depth and detail.

Common knowledge has been effectively used for the analyses of consensus in distributed computing systems (see \cref{related}). However, as traditionally defined, common knowledge is unattainable in asynchronous systems, including any platform that runs on the internet such as blockchains and cryptocurrencies, whose analysis within economics is gaining traction.\footnote{In fact, \emph{smart contracts} on blockchains are a prime example of a setting in which the rules of a mechanism are announced in a very asynchronous setting.} Our new definition allows for common knowledge to arise in such systems, and can serve as a building block in their economic analysis.

\bibliographystyle{abbrvnat}
\bibliography{ck-regained}

\end{document}